%% file: padlocks.tex
\pdfoutput=1
\documentclass{article}
\usepackage{amsmath,amsthm,amsfonts}
\usepackage{amssymb} %

\usepackage[table]{xcolor}%
\usepackage{svgcolor}

\usepackage{microtype}%

\usepackage{wrapfig}
\usepackage{array}
\usepackage{booktabs}
\usepackage{multirow}
\usepackage{paralist}
\usepackage{diagbox}
\usepackage{pict2e} %
\usepackage{makecell}
\usepackage{xspace}
\usepackage{tikz}
\usepackage{eurosym}
\usepackage{url}

\usepackage{algorithm}
\usepackage{algpseudocode}
\algnewcommand{\LeftComment}[1]{\Statex \(\triangleright\) #1}

\algblock{ParFor}{EndParFor}
\algnewcommand\algorithmicparfor{\textbf{parfor}}
\algnewcommand\algorithmicpardo{\textbf{do}}

\algnewcommand\algorithmicendparfor{\textbf{end\ parfor}}
\algrenewtext{ParFor}[1]{\algorithmicparfor\ #1\ \algorithmicpardo}
\algrenewtext{EndParFor}{\algorithmicendparfor}
\makeatletter
\renewcommand{\ALG@name}{Algorithm}
\makeatother

\algnewcommand{\IIf}[2]{\State{#1}\algorithmicif\ #2\ \algorithmicthen}
\algnewcommand{\EndIIf}{\unskip\ \algorithmicend\ \algorithmicif}

\newcommand{\N}{\ensuremath{\mathbb N}}
\newcommand{\F}{\ensuremath{\mathbb F}}
\newcommand{\PTS}[1]{\ensuremath{{\mathcal{T}}\left(#1\right)}}
\newcommand{\LO}[1]{\ensuremath{{o\mathopen{}\left({#1}\right)\mathclose{}}}\xspace}
\newcommand{\BO}[1]{\ensuremath{{O\mathopen{}\left({#1}\right)\mathclose{}}}\xspace}

\title{Optimal Threshold Padlock Systems}

\author{
Jannik Dreier\footnote{Universit\'e de Lorraine, CNRS, Inria, LORIA, F-54000 Nancy, France}
\and
Jean-Guillaume Dumas\footnote{Universit\'e Grenoble Alpes, IMAG-LJK, CNRS UMR 5224, Grenoble, France}\linebreak
\and
Pascal Lafourcade\footnote{Universit\'e Clermont Auvergne, LIMOS, CNRS UMR 6158,  Aubi\`ere, France}
\and
L\'eo Robert\footnotemark[3]
}

\usepackage{hyperref}
\hypersetup{%
        pdftitle={Optimal threshold Padlock Systems},
        pdfauthor={J. Dreier, J-G. Dumas, P. Lafourcade, L. Robert},
        breaklinks=true,
        colorlinks=true,
        linkcolor=darkred,
        citecolor=blue,
        urlcolor=darkgreen,
  naturalnames=false,
  hypertexnames=false,
 }
\usepackage[capitalise,nameinlink]{cleveref}

\usepackage{thmtools, thm-restate}
\declaretheorem{theorem}
\declaretheorem[sibling=theorem]{lemma}

\declaretheorem[sibling=theorem]{proposition}
\declaretheorem[sibling=theorem]{definition}
\declaretheorem[sibling=theorem]{definitions}
\declaretheorem[sibling=theorem]{remark}

\begin{document}
\maketitle

\begin{abstract}
\input{abstract}
\end{abstract}

\section{Introduction}

\input{intro}

\input{threshold}

\input{solution}

\input{minimality}
\input{algebra}
\input{sqrt}

\input{lk11_zones}

\input{cex}
\input{recursive}

\section{Secret sharing with reduced field size}\label{sec:fieldsize}
Consider Shamir's secret sharing via interpolation over a finite
field. For a secret value within a finite field $\F_q$, set it
as the evaluation at zero of a degree $k-1$ polynomial whose other
coefficient are randomly sampled. Then
distribute an evaluation of the polynomial at distinct non-zero points
to $n$ participants. This is a $k$-out-of-$n$ threshold system.
It requires that there are enough evaluation points for all the
participants and thus that $q>n$.

We in fact have shown that this is optimal in certain cases, but that one can
use smaller fields in others: instead of a degree $k$ polynomial, use
a degree $t$ polynomial, where $t$ is the number of padlocks in one of
our systems.
This number of padlocks $t$ is in fact the number of available
evaluation points. Then the identical keys for a given padlock are the
evaluations of the polynomial at the points. Thus participants have
several evaluations instead of a single one.
We have thus proposed a $k$-out-of-$n$ secret sharing scheme where the
field size is reduced. For instance, from \cref{thm:recursive}, if
$n\geq{2k}$, then it is sufficient to take the field size $q=\BO{\log(n)^{k-1}}$.

\section{Conclusion}\label{sec:conclusion}

\input{conclusion}

\bibliographystyle{plainurl}
\bibliography{biblio}

\end{document}

%% file: abstract.tex
In 1968, Liu described the problem of securing documents in a shared
secret project. In an example, at least six out of eleven
participating scientists need to be present to open the lock securing
the secret documents. Shamir proposed a mathematical solution to this
physical problem in 1979, by designing an efficient $k$-out-of-$n$
secret sharing scheme based on Lagrange's interpolation. Liu and
Shamir also claimed that the minimal solution using physical locks
is clearly impractical and exponential in the number of participants.
In this paper we relax some implicit assumptions in their claim and
propose an optimal physical solution to the problem of Liu that uses
physical padlocks, but the number of padlocks is not greater than the
number of participants.  Then, we show that no device can do better
for $k$-out-of-$n$ threshold padlock systems as soon as
$k\geq{\sqrt{2n}}$, which holds true in particular for Liu's example.
More generally, we derive bounds required to implement any threshold
system and prove a lower bound of $\BO{\log(n)}$ padlocks
for any threshold larger than $2$. For instance we propose an optimal
scheme reaching that bound for $2$-out-of-$n$ threshold
systems and requiring less than $2\log_2(n)$ padlocks.
We also discuss more complex access structures, 
a wrapping technique, and other
sublinear realizations like an algorithm to generate $3$-out-of-$n$
systems with $2.5\sqrt{n}$ padlocks.
Finally we give an algorithm building $k$-out-of-$n$ threshold padlock
systems with only $\BO{\log(n)^{k-1}}$ padlocks.
Apart from the physical world, our results also show that it is
possible to implement secret sharing over small fields.

%% file: intro.tex
In 1979, in his paper on secret
sharing~\cite{DBLP:journals/cacm/Shamir79}, A.~Shamir presented the
following threshold problem introduced by C.~L.~Liu in~\cite[Example
1-11]{nla.cat-vn2101463}: \emph{\small Eleven scientists are working
on a secret project. They wish to lock up the documents in a cabinet
so that the cabinet can be opened if and only if six or more of the
scientists are present. What is the smallest number of locks needed?
What is the smallest number of keys to the locks each scientist must carry?}
Liu and Shamir answered this \emph{physical} problem using mathematics as
follows: \emph{\small It is not hard to show that the minimal solution uses
  462 locks and 252 keys per scientist. These numbers are clearly
  impractical, and they become exponentially worse when the number of
  scientists increases.}
This is why Shamir proposed to use polynomial and Lagrange's interpolation
to solve Liu's question. His clever idea is to hide the secret in the
constant term of a polynomial of degree $k{-}1$. Then he distributes one
point of the chosen polynomial to each of the $n$ participants. As
soon as $k$ participants share their points, they can recover the
secret using Lagrange's interpolation and algorithms in
$O(n\log^2 n)$ operations~\cite{10.5555/578775,10.5555/270146}. A few
years later, verifiable secret sharing was introduced by Chor et
al. in~\cite{10.1109/SFCS.1985.64} and improved
in~\cite{10.1109/SFCS.1987.4}. The idea is to offer the possibility to
verify if the points are valid.

We show that Liu's problem is solvable using far less locks.
Liu and Shamir claim stems from the restriction that there
should be a lock for each combination of $6$ scientists,
$462=\binom{11}{6}$, and that every scientist needs the keys for every
combination of scientists that includes him. This is
$252=\binom{10}{5}$ keys. Liu-Shamir's minimality result thus assumes
that the only physical arrangements of locks that allow threshold
systems are those where the opening of any lock opens the cabinet.

{\bf Contributions:}
\begin{enumerate}
\item As a warm up, we relax Liu and Shamir's assumption
and design a \emph{physical $k$-out-of-$n$ threshold padlock system}.
We have
build a prototype of this physical device.
Our system only requires one padlock and one key per
participant, which is practical, when compared to the previous
exponential solution.
 \item Then, we establish \emph{lower bounds} on the number of padlocks necessary
for any abstract threshold system.

Specifically, we show that for a $2$-out-of-$n$ configuration,
less than $2\lceil\log_2(n)\rceil$ padlocks are sufficient, provided that keys can
be duplicated. In fact, there is an optimal solution for this type of
configuration, with ${\mathcal O}(\log(n))$ padlocks, and we also show
that this optimum can be realized, using our physical system as one
building block.

Differently, for $k$-out-of-$n$ configurations with $k\geq{3}$, it is
more complicated to solve the problem with fewer than $n$ padlocks.
We first prove that this is impossible for $k\geq{\sqrt{2n}}$ and thus
that our physical device is optimal in these cases.
For instance, this answers Liu's question: the minimal number of
padlocks for a $6$-out-of-$11$ configuration is $11$, as
$6>\sqrt{22}$. Our system with $11$ padlocks and only $1$ key per
participant, is thus optimal in this case.
We are then nonetheless able to give algorithms
building systems for $k=3$ with only about $2.5\sqrt{n}$ padlocks and
each participant owns $3$ keys. These realizations use more complex
access structures and associated algorithmic building blocks,
that we provide.

\item We discuss more \emph{complex access structures}, which include for instance ensuring that Alice and Bob
can open the lock with any other third participant, but not
together. Another possibility is for instance that Alice is highly
ranked and can open the padlock by herself but that any others need to
be at least two.
For this we develop a tentative
\emph{padlock algebra} for logic gates and give associated algorithms.
The idea is to combine threshold cryptography and
secret sharing with the theory of block designs, packings and Sperner
families.

\item Finally, we propose a \emph{recursive algorithm} to build larger systems,
that requires only a logarithmic number of padlocks.
Asymptotically, our algorithm requires only $\BO{\log(n)^{k-1}}$
padlocks to realize a $k$-out-of-$n$ threshold padlock system.

\item Lastly, we also show that our physical results do apply to the \emph{numerical world}
by linking the number of padlocks to the size of the finite field used
for secret sharing.
\end{enumerate}

{\bf Outline:} In \cref{sec:existthreshold}, we review existing
threshold mechanisms that use physical padlocks, or visual
cryptography, which not perfectly answer Liu's problem.  In
\cref{sec:our}, we describe our novel physical $k$-out-of-$n$
threshold padlock system device.  In \cref{sec:minimal}, we derive
generic bounds on the number of padlocks required to realize a given
threshold configuration.  We also show the optimality of our device
for $2$-out-of-$n$ systems.  Then, in \cref{sec:algebra}, we discuss
more complex access structures.  We provide for instance solutions on
logic formulae in \cref{ssec:normalforms,app:blocks}, and in
\cref{ssec:knotted} we introduce the use of a sealed wire.  Further
lower bounds, including the optimality of our solution for
sufficiently large $k$, together with smaller realizations, with
strictly less than $n$ padlocks, are given in \cref{sec:sqrt}.  Our
recursive construction is then given in \cref{sec:rec} and the link
with the numerical aspects in \cref{sec:fieldsize}.

%% file: threshold.tex
\section{Related Work}\label{sec:existthreshold}

Threshold cryptography in general received a lot of attention
recently, since on March 1, 2019 the Computer Security Division (CSD)
at the National Institute of Standards and Technology
(NIST)
published the final version of NISTIR
8214, ``\emph{Threshold Schemes for Cryptographic
  Primitives}''~\cite{nist:threshold}.
This reports explicitly also mentions physical threshold solutions (page
10, line 55): ``\emph{While we focus on secure implementations of
  cryptographic primitives, the actual threshold techniques may also
  include non-cryptographic techniques}.''  We present existing
physical solutions for threshold cryptography, while a survey of
cryptographic threshold schemes by Y.~Desmedt can be found
in~\cite{Desmedt2011}.  We distinguish two classes of solutions: the
first one uses physical keys and padlocks; the second one uses visual
cryptography, as introduced by M.~Naor and A.~Shamir in
1994~\cite{DBLP:conf/eurocrypt/NaorS94}.

\subsection{State of the art, using Padlocks}

A $1$-out-of-$1$ padlock is just one simple
physical padlock. There are many systems for $1$-out-of-$n$ padlocks, both
home made and commercial products.  There also exist commercial
solutions for $n$-out-of-$n$ padlocks, which are used by for example by
electricians to secure an electrical circuit as explained next.

\subsubsection[1-out-of-n padlocks]{$1$-out-of-$n$ locks}\label{sec:1on}
In \cref{fig:lock1ondaisy}, left, a $1$-out-of-$2$ padlocks is done
simply with two physical padlocks.  This approach can be generalized to
$1$-out-of-$n$ as in \cref{fig:lock1ondaisy}, right, and is
called a \emph{daisy chain}. We notice that the bottom left yellow
padlock was badly placed, and it is useless. In this case the owner of
this padlock cannot open the door.  We call this the \emph{daisy chain
  attack}.  For example in \cref{fig:lock1ondaisy}, if the owner
of the bottom padlock opens it and then locks it upper in the chain, then
he excludes all the owners of these padlocks, as they cannot open the
door any more\footnote{A deliberate attack adding an
  additional chain and padlock to the gate, or even welding padlocks
  together, is always possible, and out of scope here: we aim to protect
  against attacks that could be ``excused'' with a wrong use 
  of the system.}.
\begin{figure}[htbp]
\centering
\includegraphics[height=4cm]{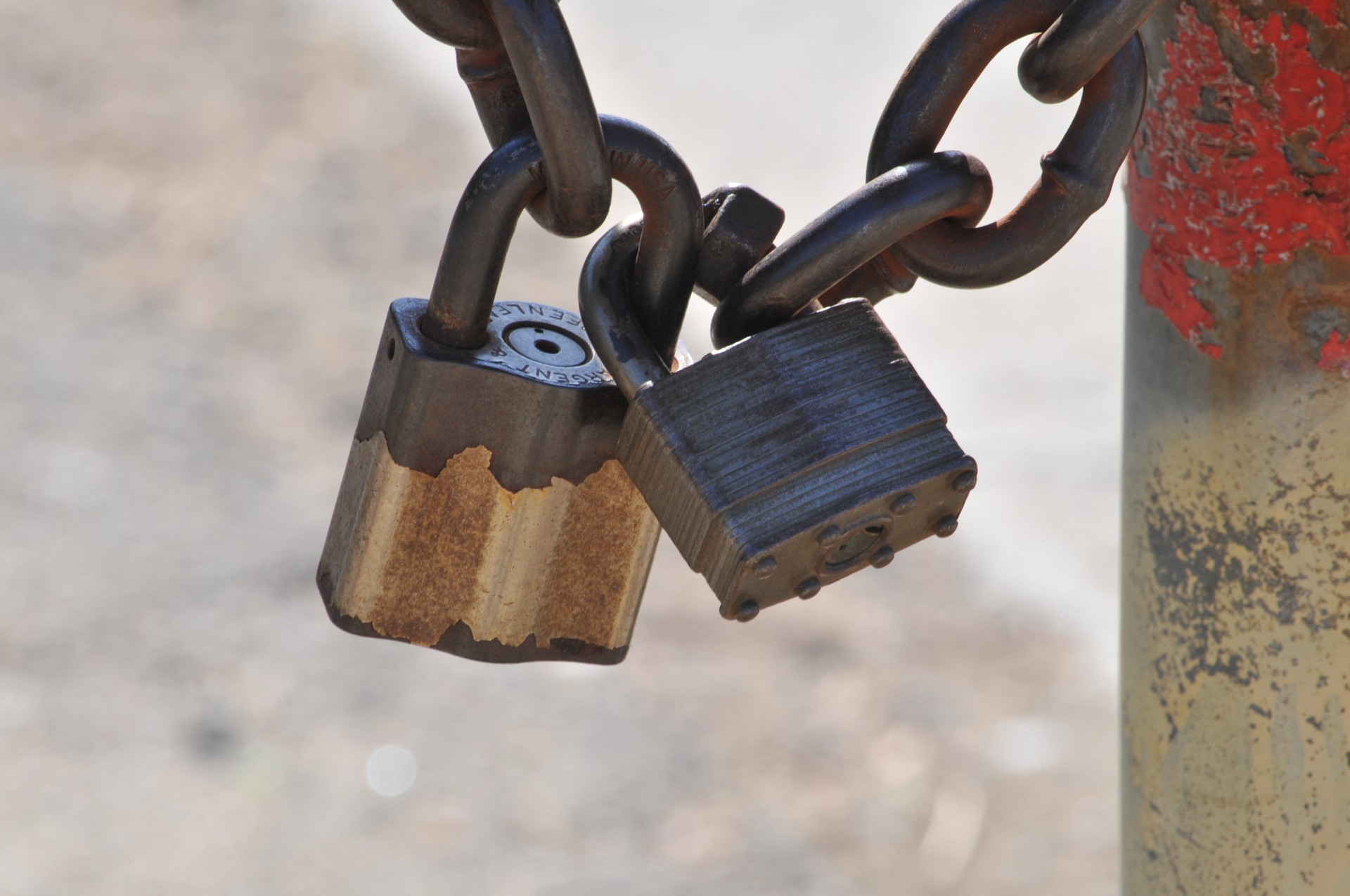}%
~~%
\includegraphics[height=4cm]{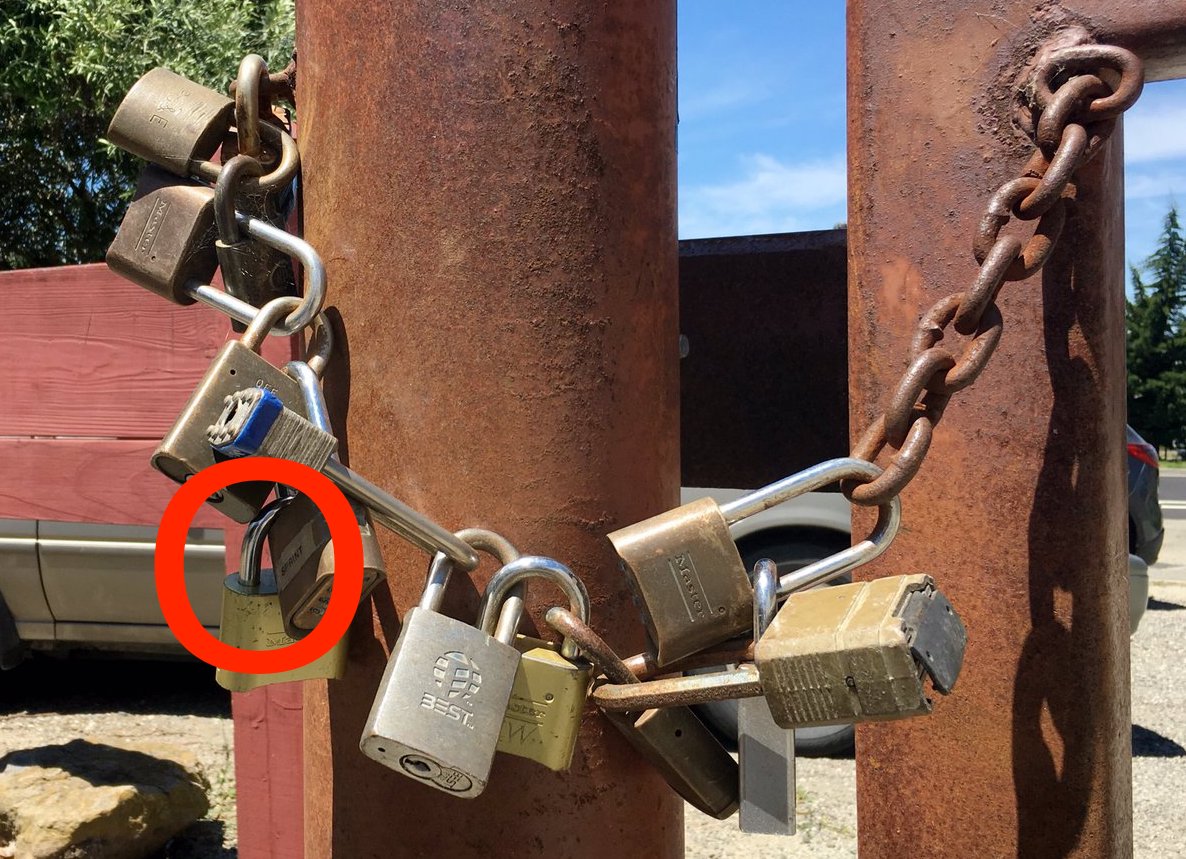}
\caption{Physical $1$-out-of-$n$ padlocks, forming a daisy
  chain. {\bf Left:} The simplest daisy chain with two padlocks. {\bf Right:} A longer daisy chain wiht one useless padlock.}\label{fig:lock1ondaisy}
\end{figure}

In \cref{fig:lock1o6}, we can see two different mechanisms that
perform $1$-out-of-$6$ padlocks to open the gate of a field. The first
one has six padlocks that block the trigger. As soon as one padlock
is opened a latch is removed and then the door can be opened.  It is
the natural extension of the solution of \cref{fig:lock1ondaisy}
that avoids the daisy chain attacks.
Next, the second picture of \cref{fig:lock1o6}, shows a different
solution also implementing a
$1$-out-of-$6$ padlock, and which is also resistant to the daisy chain
attack. In this system, as soon as one padlock is removed, it is possible
to turn the circle and then to pass the stick in the corresponding
hole in order to open the door.

\begin{figure}[htb]
\centering
\includegraphics[height=5cm]{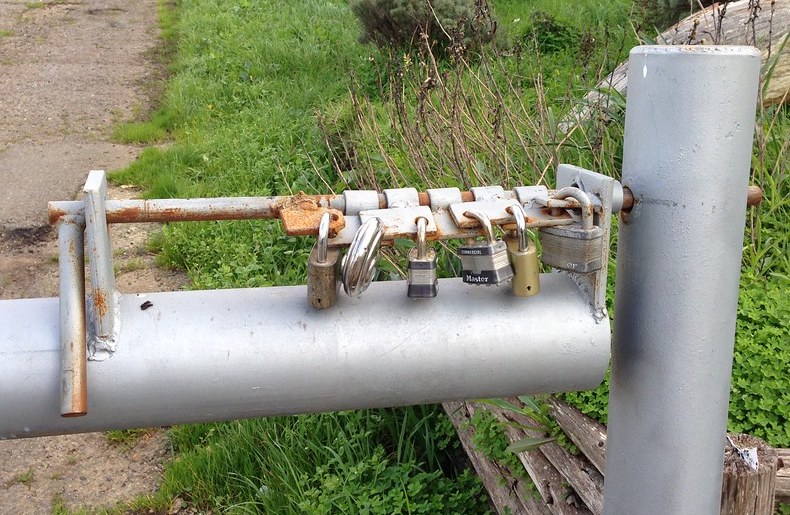}%
~~%
\includegraphics[height=5cm]{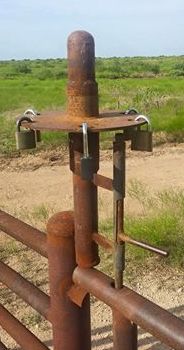}%
\caption{Two ad-hoc physical $1$-out-of-$6$ padlocks.}\label{fig:lock1o6}
\end{figure}
\begin{figure}[htb]
\centering
\includegraphics[height=3cm]{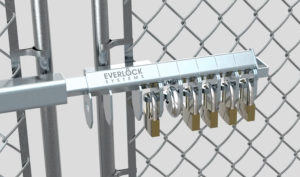}%
~~%
\includegraphics[height=3cm]{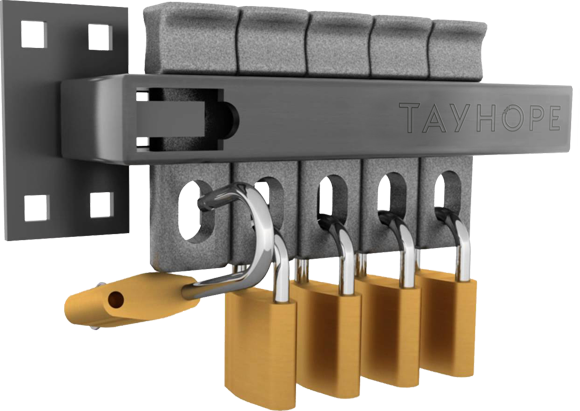}%
\caption{
Physical $1$-out-of-$5$ padlocks, first by Everlock System, then
by Tayhope multi-locking system.}
\label{fig:lock1o10T}
\end{figure}

There are also commercial products for $1$-out-of-$n$ padlocks. 
The first picture of \cref{fig:lock1o10T} shows a commercial
product designed 
by Everlock Systems: the model SLX2~\cite{Everlock:SLX2}.
The second picture of \cref{fig:lock1o10T} shows a commercial product
sold by Tayhope Multi-Locking Systems~\cite{Tayhope:multilock}.
Everlock Systems has multiple
patents on 
their designs~\cite{everlock1,everlock2,everlock3,everlock4} and their
solution is close to the mechanism proposed
on the left side of \cref{fig:lock1o6}. 
Differently, Tayhope mechanism allows the owner of a padlock to remove the metallic
stick which enables the opening of the door, by pushing all the padlocks
on one side.

Now, if one is interested in reducing the number of padlocks, one can
realize a $1$-out-of-$n$ threshold system with a single lock:
duplicate the key of one padlock $n$ times and distribute the key to
all the participants. The obtained system has not all the physical
properties of the daisy chain or the systems of
\cref{fig:lock1o10T} (for instance the latter does not need a
trusted third party to setup the chain or to duplicate the keys), but
is probably more economical.
Overall, we have the following possibilities for $1$-out-of-$n$
systems:
\begin{itemize}
\item A single padlock with $n$ duplicated keys: probably most economical;
\item A daisy chain: if keys cannot be duplicated;
\item Systems like those of \cref{fig:lock1o10T}: they do not
  require a trusted third party for the setup, as each participant can bring their own lock and key(s).
\end{itemize}

\subsubsection[n-out-of-n locks]{$n$-out-of-$n$ locks}
Finally, there are physical $n$-out-of-$n$ mechanisms using padlocks that
are used for example for operations on high-voltage circuits and
transformers.  Two examples of $6$-out-of-$6$ padlocks are given in
\cref{fig:lock6o6}. The idea is that nobody should be able to
turn on the electricity while someone is still working on the
high-voltage transformer. To achieve this, each technician places a
padlock on the main switch before entering the danger zone.  This ensures
that all technicians have to leave the danger zone before electricity
can be restored. The example can easily be extended to a
$n$-out-of-$n$ system.

\begin{figure}[htbp]
\centering
\includegraphics[height=3cm]{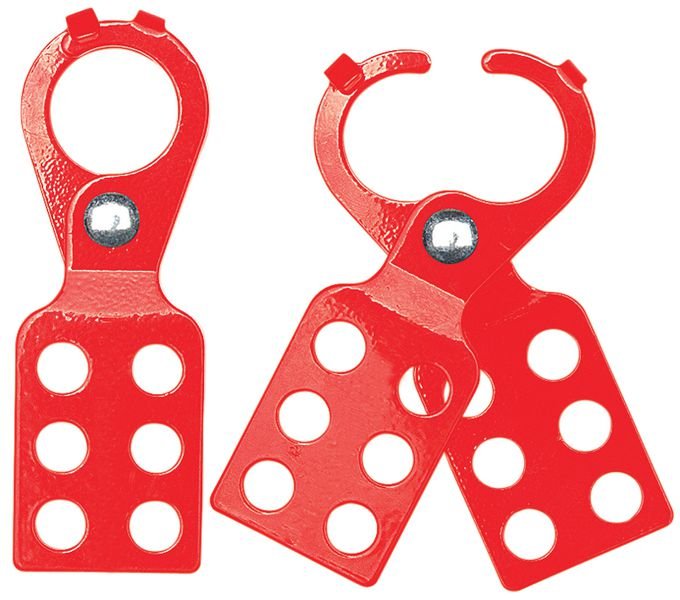}%
~~\includegraphics[height=3cm]{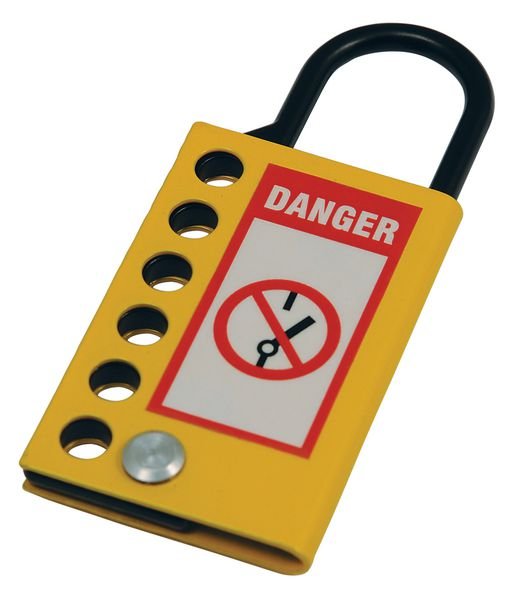}
~~\includegraphics[height=3cm]{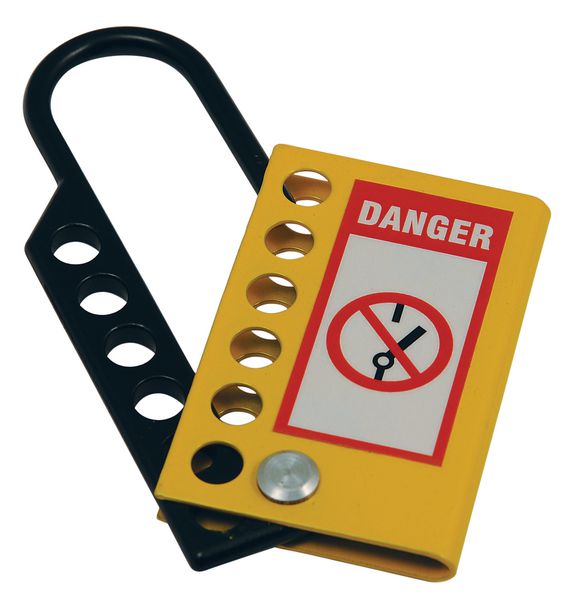}
\caption{Physical 6-out-of-6 padlocks, 
  by Seton (models SLECO and MANM8).}\label{fig:lock6o6}
\end{figure}

\subsection{Using Visual Cryptography}

In 1994, M.~Naor and A.~Shamir proposed the \emph{visual
  cryptography}~\cite{DBLP:conf/eurocrypt/NaorS94,DBLP:conf/spw/NaorS96}
for black and white images.  This was improved
in~\cite{DBLP:journals/ipl/BlundoSN00} for gray images and
in~\cite{HOU20031619} for color images.

The idea is to split a secret
into two images printed on transparent paper in a way such that their
superposition makes the secret appear.  An example is given in
\cref{fig:visual}.

\begin{figure}[htbp]
\centering
    \includegraphics[height=1.8cm]{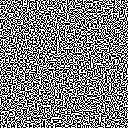}
    \includegraphics[height=1.8cm]{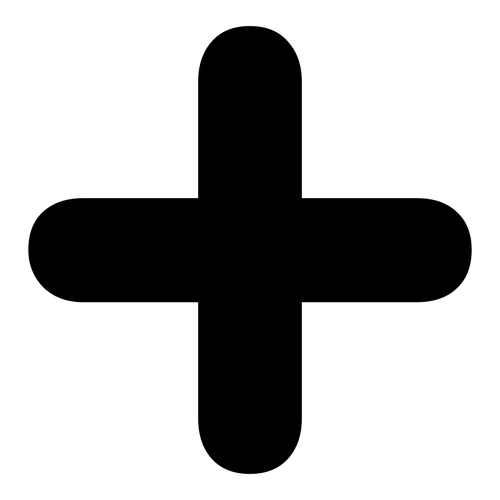}
    \includegraphics[height=1.8cm]{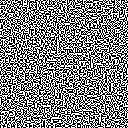}
    \includegraphics[height=1.8cm]{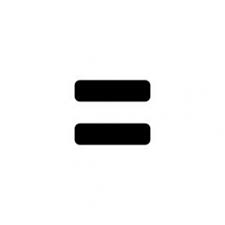}
    \includegraphics[height=1.8cm]{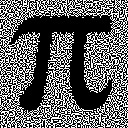}
\caption{Example of visual cryptography, superposing the images let
  the symbol $\pi$ appear.}\label{fig:visual}
\end{figure}

For color images, security cannot be perfectly achieved for more
than 3 colors~\cite{LEUNG2009929}.
In~\cite{10.1023/A:1008280705142}, the authors proposed a
generalization of the approach to $k$-out-of-$n$ images.  This can be
used as a first physical answer to Liu's problem. This solution is not
really practical since it needs a computer to compute the different
images. Moreover in~\cite{10.1007/s10623-005-6342-0}, the authors show
that it is possible to cheat in visual cryptography by introducing
fake shares that change the result. This clearly shows that this
solution is not verifiable, which requires the ability to check that
shares are valid.

%% file: solution.tex
\section[A Novel Physical k out n Threshold Lock]{A Novel Physical $k$
  out $n$ Weighted Threshold Lock}\label{sec:our}

As a natural extension of $1$-out-of-$n$ systems, we design a $k$-out-of-$n$ physical threshold lock that uses $n$ padlocks and works as follows.
Each padlock secures one block, with a latch, attached to a sliding bar, which is limited in its sliding movement by the blocks.
If sufficiently many blocks are removed, the sliding bar can be moved far enough to open the barrier.
In Algorithm~\ref{alg:koutofnphys}, we describe our solution in a
generic way. 

We also have built a wooden prototype that can be configured for
different cases, see \cref{fig:locknon} for a 2-out-of-3
configuration. 

\begin{figure}[htb]
  \begin{center}
\includegraphics[height=3.25cm]{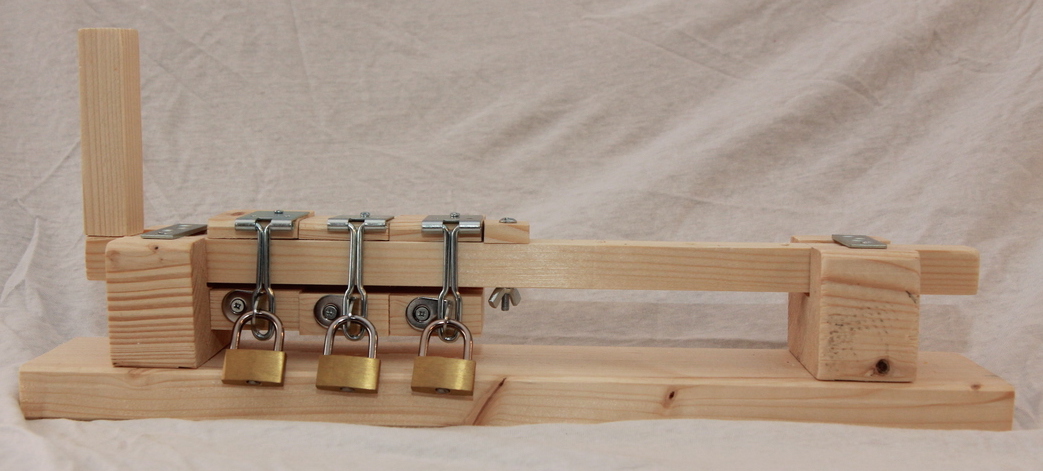}
\\[4pt]
\includegraphics[width=.45\textwidth]{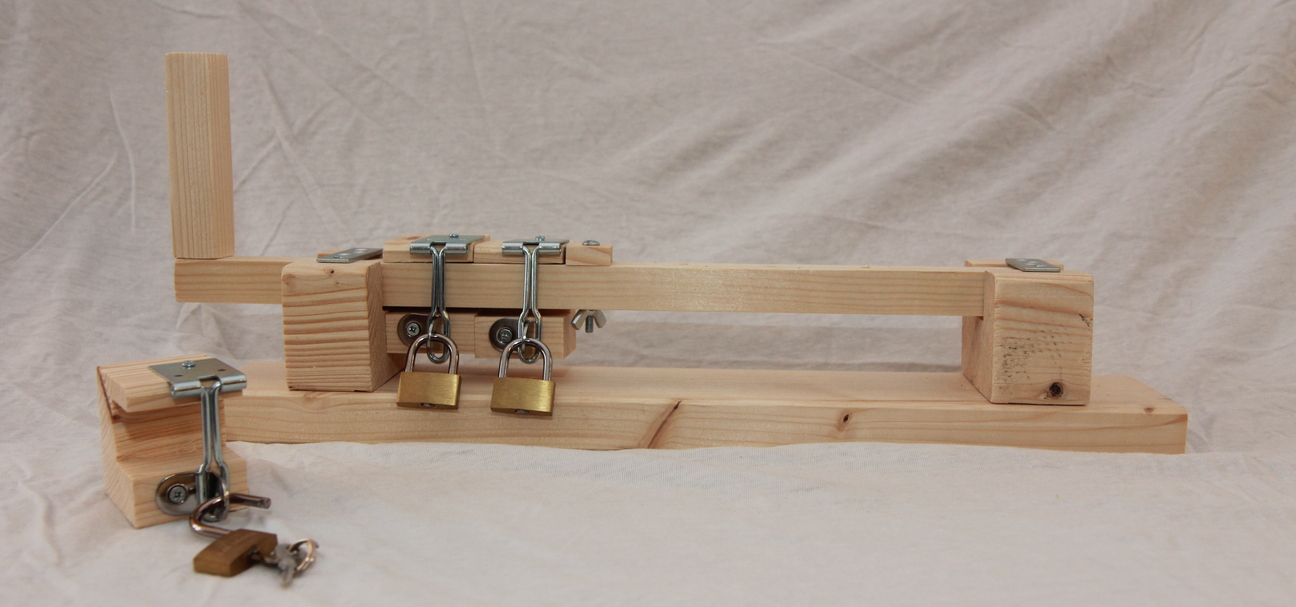}
\includegraphics[width=.45\textwidth]{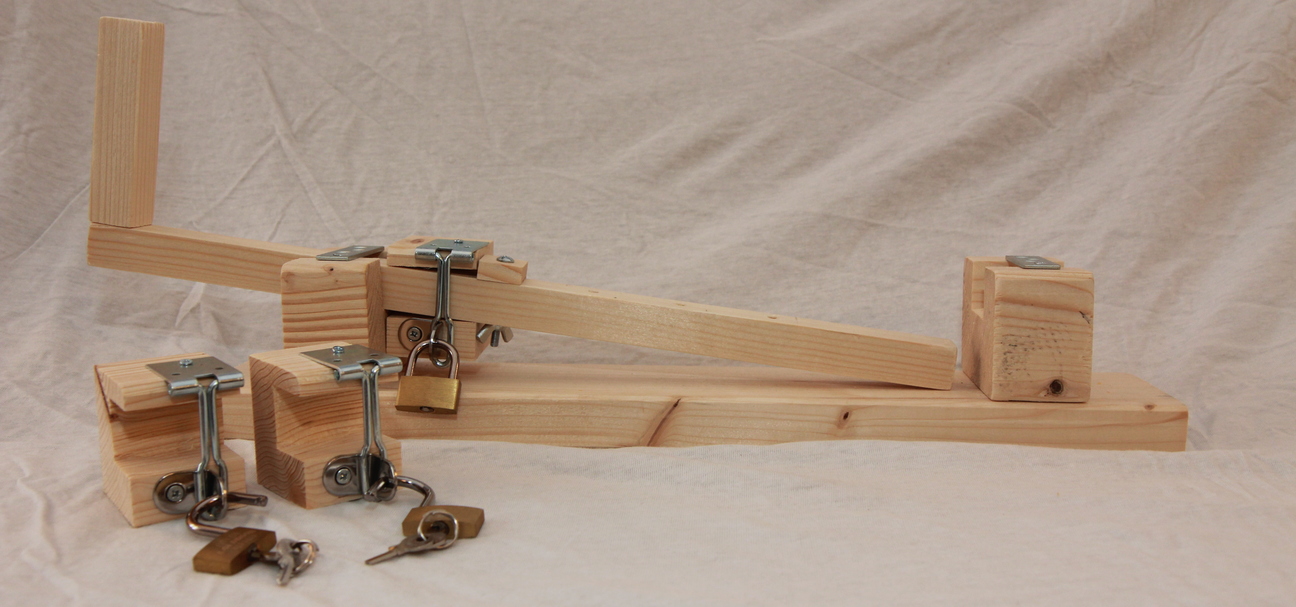}
\end{center}
\caption{Physical $2$-out-of-$3$ lock. {\bf Top:} all three blocks
  attached, padllock closed. {\bf Left:} one block removed, the bar can
  be moved to the left, but not sufficiently far to open. {\bf Right:}
  two padlocks are removed, the barrier is open.}\label{fig:locknon}
\end{figure}

In this example, on the top image, we have $3$ padlocks
attached to support of size $l$ and the blocker is installed just
after them on the initial configuration. The bar is installed in such
a way that it over passes slightly more than the size of one padlock
support on the right.  
In the left image, once one padlock and its support is removed
then the bar can move to the left but not totally be removed. 
Finally, once two padlock supports are removed we can open the system.

\begin{algorithm}[htb]
\caption{$k$-out-of-$n$ physical threshold lock}\label{alg:koutofnphys}
\begin{algorithmic}[1]
\Require $k\leq{n}$, $n$ padlocks, $n$ supports in wood of the same size
$l$, a  bar of wood of size at least $(n+k-1)\times l$, a support of size
$l$, and a blocker.
\Ensure A $k$-out-of-$n$ physical threshold padlock.
\State Distribute one individual key of one padlock to each of the $n$ participants.
\State Lock the $n$ padlocks on the $n$ latches attached to wooden
supports on the long wooden bar.
\State Install the blocker after the $n$ padlocks.
\State Install the bar in the system in order that it over passes by
slightly more than $(k-1)\times l$ (and not more than $kl$) the system general lock.
\end{algorithmic}
\end{algorithm}

\input{prototype}

Our technique can also be used to implement \emph{weights} by using
blocks of different sizes.
\cref{fig:lockweight} shows an example where either one ``master'' key (opening the padlock on a larger block) can be used to open the lock, or any two of the other keys (opening the padlocks on the smaller blocks).
The same idea can also be used to implement a policy where, e.g., either Alice and one other participant, or any three other participants are required to open the lock.
It suffices to give Alice the keys for the larger block, and use a configuration that requires the removal of three small blocks to open.

\begin{figure}[htbp]
  \begin{center}
\includegraphics[height=3.25cm]{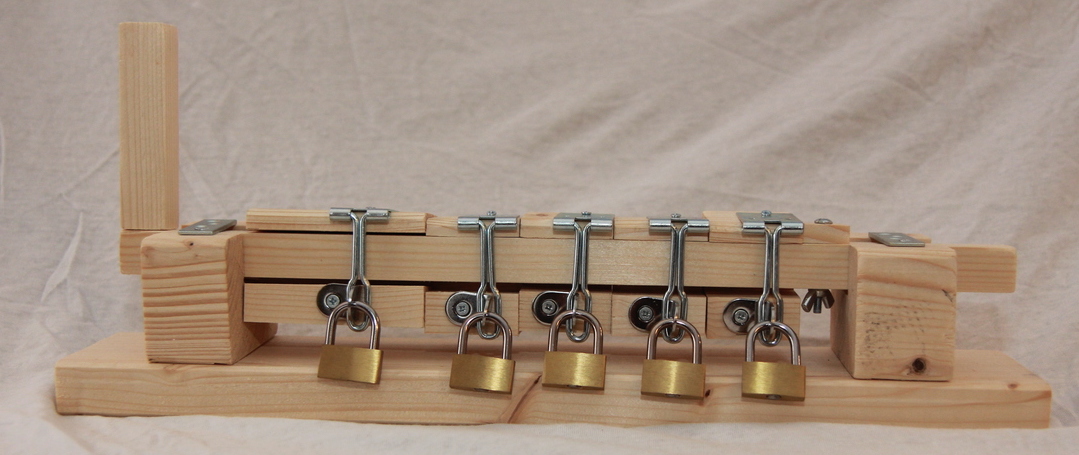}
\\[4pt]
\includegraphics[width=.45\textwidth]{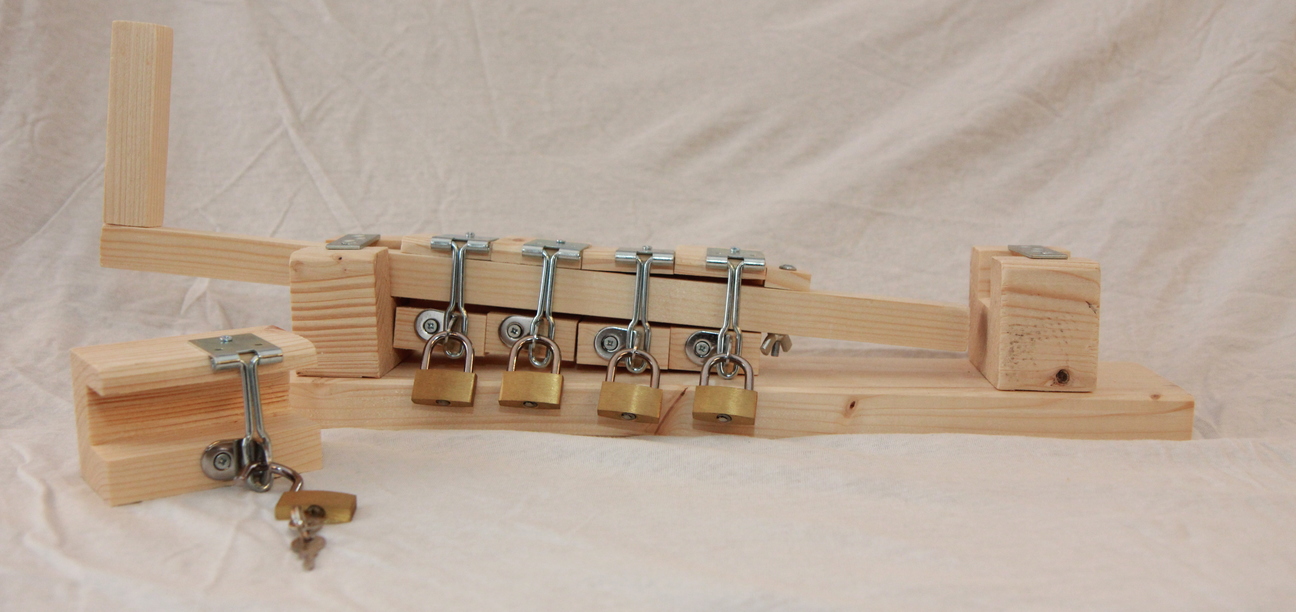}
\includegraphics[width=.45\textwidth]{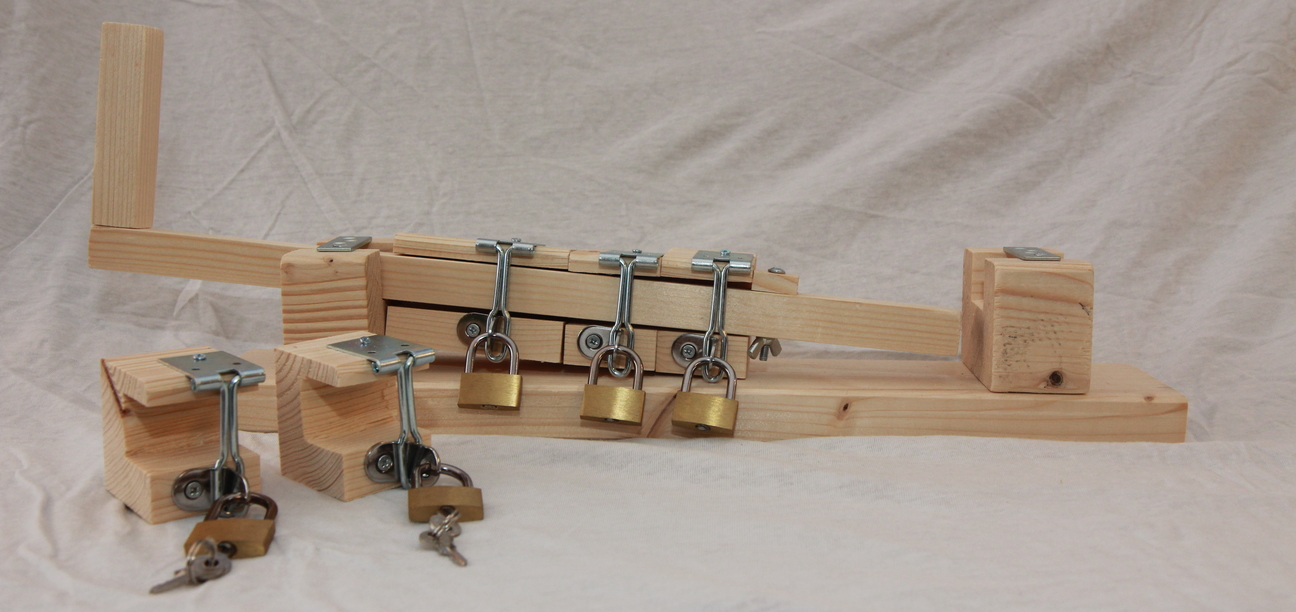}
\end{center}
\caption{Physical lock with weighted keys. {\bf Left:} all blocks are
  attached and locked. {\bf Right:} the ``master'' key was used to
  open the larger block, unlocking the barrier. {\bf Bottom:} two
  ``normal'' keys opening normal blocks also allow to open the
  barrier.}\label{fig:lockweight}
\end{figure}

Our system is \emph{ad-hoc} since once it is set up, each participant can install their own
lock, which avoids having to trust the dealer as in existing
cryptographic solutions.
Note that to avoid problems during the setup phase, we assume that all participants install their locks at the same time, right after the lock has been set up.

Our solution is also \emph{reusable} as it can be locked again, unlike
for example a solution using cryptographic secret sharing to share a
code for a combination lock, where the code would be revealed once and
for all: such a lock thus cannot be effectively locked again without
changing the code. Note that a system with a combination lock would
also require a special procedure or a trusted third party to setup the
combination initially.  Moreover, our system also protects users
against the \emph{daisy chain attack} as only one padlock can be
fitted to the latch of any block.

By construction our solution is \emph{verifiable} since everyone can
check if there is at least one padlock that can be opened with the secret
key that he has received. Comparing to the mathematical solution
proposed in~\cite{10.1109/SFCS.1987.4} consisting in giving extra
information to each participant to convince him that he received a
valid point of the polynomial, our solution does not require any extra
material, nor does it require any trusted third party.  There are thus at least three direct applications of our
physical threshold system:%
\begin{enumerate}
\item Our system can be used to construct a physical verifiable secret
  sharing protocol. As it can also easily be extended to deal with
  weights, we also have a physical equivalent to the cryptographic
  protocol given in~\cite{10.1007/978-3-540-30576-7_32,6138912}.
\item Threshold cryptography has been applied to voting, e.g.,
  in~\cite{Schoenmakers99asimple}. Our system can be used to secure
  physical pen and paper voting, by ensuring that the ballot box can
  only be opened if $k$-out-of-$n$ trustees agree.
\item As a user never has to reveal his physical key, our
mechanism can also be used to design a $k$-out-of-$n$ authentication
mechanism.
\end{enumerate}

%% file: prototype.tex
\cref{fig:config}, left, shows our prototype in a $2$-out-of-$4$ configuration.
The prototype can be configured for $k$-out-of-$n$ systems for any $k \in \{1, 2, 3\}$ and $n \in \{3, 4, 5, 6\}$.
By moving the wooden block attached to the moving bar (red circle in
\cref{fig:config}, right) one can fix the number of blocks that can be attached, i.e., $n$.
By moving the block on the right (blue circle in
\cref{fig:config}, right) one can fix the number of blocks that
need to be removed before the bar can be opened, i.e., the threshold
$k$: on the rightmost position, removing one block is sufficient to
open the bar.
When moving this block to the left one can increase the number of
blocks that need to be removed before the bar opens.
For convenience, in our prototype everything can be easily adjusted
using screws, but obviously, in a real implementation, they need to be
permanently fixed to ensure security.

\begin{figure}[htbp]
  \begin{center}
\includegraphics[height=3.25cm]{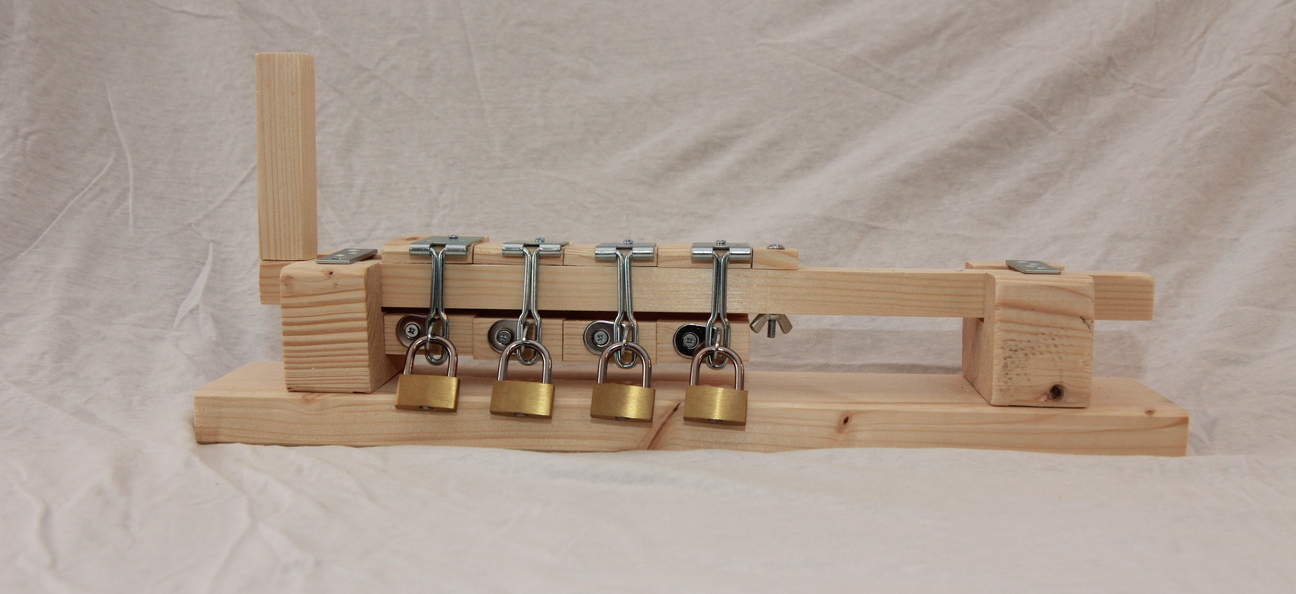}
~~
\includegraphics[height=3.25cm]{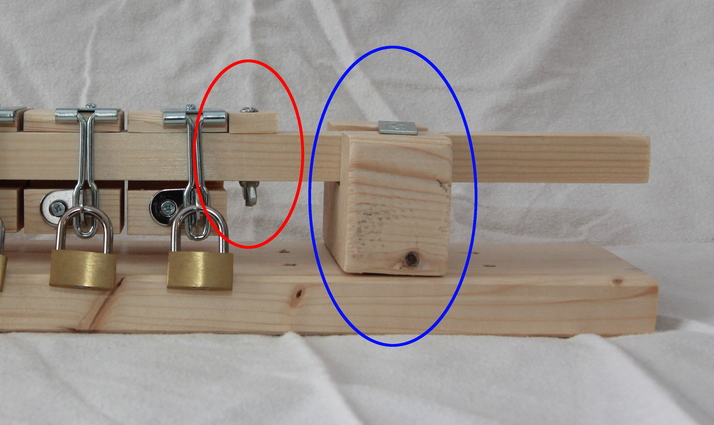}
\end{center}
\caption{{\bf Left:} Physical $2$-out-of-$4$ threshold lock system;
  {\bf Right:} configuring our prototype: the block in the red circle fixes
  the number $n$ (here: 4), the block in the blue circle fixes the
  threshold $k$ (here: 3). One can also see the holes that allow the
  block to be fixed in other positions.}\label{fig:config}
\end{figure}

%% file: minimality.tex
\section{Formalization and Generic Bounds on the Number of Padlocks}\label{sec:minimal}

We now establish bounds on the number of padlocks required to realize
a certain threshold. %
We assume that padlocks are more expensive than keys, i.e., we will
try to implement threshold systems with fewer padlocks, even if this
means duplicating some of the keys.  We define a padlock system to be
any arrangement of padlocks protecting something. For the sake of
simplicity, in the following, we consider this to be the possibility
to ``\emph{open a door}''.

\begin{definitions}\label{def:padlocksystem}
  A \emph{padlock} is a device requiring a single \emph{key} to be
  opened (keys can be duplicated).
  A \emph{padlock-system} is a device comprising an arrangement of
  \emph{latches} that prevents a \emph{door} to be opened when some
  padlocks are attached to some of the latches.
\end{definitions}

\begin{definition}\label{def:thresholdsystem}
  A $k$-\emph{threshold padlock system} is a padlock system with an
  arrangement of padlocks and a distribution of keys that allows any
  group of $k$ or more participants to open the door and prevents any
  group of strictly less than $k$ participants to open it.
\end{definition}

\begin{remark}\label{rem:digital}
  While directly applicable to physical padlock systems, this
  definition also applies to some cryptosystems.  For instance
  consider any symmetric or asymmetric cryptosystem with a shared
  (duplicated) decryption key. Closing a padlock could just be
  ciphering with an encryption key; setting a padlock-system could
  just be multiple encryption (even if electronic threshold
  cryptosystems are more complicated) and opening the door is
  deciphering. For this example, the only difference with physical
  system is that the order of encryption must be taken into account
  for decryption.

  Now, most of the \emph{lower} bounds described in this section only
  suppose the existence of a threshold system satisfying the above
  definitions.  Therefore those lower bounds also apply to electronic
  threshold cryptosystem satisfying
  \cref{def:padlocksystem,def:thresholdsystem}.
\end{remark}

\begin{definition} Let $n$ be the number of players and $k\leq n$ be a
  threshold of players required to open the ``door''.  Then
  $\ell_{k,n}$ is the minimal number of padlocks, in any arrangement,
  allowing a $k$-out-of-$n$ threshold opening of the door.  Also, we
  define the \emph{rank} of an arrangement of padlocks and keys as the
  maximal number of keys owned by any player.
\end{definition}

For instance, we have that:
\begin{itemize}
\item $\ell_{1,n}=1$: one padlock with everybody having a copy of the
  same key is sufficient.
\item $\ell_{k,n}\leq n$: by our system described in
  \cref{sec:our}, see~\cref{fig:locknon}.
\end{itemize}

\subsection{Sperner Families}
Using the fact that all subsets of size $k$ of the $n$ participants
can open the door, and no subset of $k-1$ or less can do it, we have
the following results.  First, it is easy to see that with only $k-1$
or fewer different locks, one cannot ensure a threshold of at least
$k$.
\begin{restatable}{lemma}{lemktwo}\label{lem:ktwo}
$\forall k\geq 2, \ell_{k,n}\geq k$.
\end{restatable}

\begin{proof}
Suppose for $\ell_{k,n}$ we have an existing threshold system where a minimum
of $k$ people is required to open the door, and moreover any subset of $k$
people can open it.
Suppose $t=\ell_{k,n}\leq k-1$ and consider one group of $k$ people
able to open the door.
For this, whatever the arrangement, they had to open some of the $t$ padlocks,
thus with at most $t$ keys.
This is less keys than the number of people, so there must exist a
subgroup of at most $t$ people owning these $t$ keys and:
\begin{itemize}
\item Any of the $k$ people must own at least one of the $t$ keys, otherwise
  they are not required to open the cabinet and $k-1$ people are enough.
\item By induction on a subgroup of size $1\leq{u}<t$, an $(u+1)$-th person,
  among the remaining $k-u$, must own the key of one padlock not owned
  by the previous $u$, otherwise this person is not required and $k-1$
  people are enough.
\end{itemize}
Now, this subgroup of size at most $t$ is thus able to open the door
by themselves. But $t \leq k-1 < k$, is below the threshold, a
contradiction.
\end{proof}

Second, we see that if the set of keys of a participant is included in another
participant's set of keys, intuitively the first participant is
``useless'' to achieve the threshold.
\begin{restatable}{lemma}{lemnosubsets}\label{lem:nosubsets}
  Let $k\geq 2$, and set up an arrangement of padlocks and a
  distribution of keys with a $k$-out-of-$n$ threshold opening.  No
  participant can own a set of keys that is a subset of another
  participant's set of keys.
\end{restatable}

\begin{proof}
Let $A$ have a set of keys included in that of $B$. As $k\geq 2$, $A$
and $B$ can be in a size $k$ subset of participants that can open the
door. But then the keys of $A$ are useless since $B$ has all of them.
Therefore there would be a size $k-1$ subset of participants able to
open the door, a contradiction.
\end{proof}

This shows for instance that each participant must have at least one
key.  Further, this means that the sets of keys must form a
family of inclusion-free subsets.  This is called a \emph{Sperner
  family} or a \emph{clutter}~\cite{Sperner:1928:clutter}.  The
padlocks can then be seen as the vertices of a hypergraph, where each
participant is represented by a hyperedge, the set of its owned
keys. The \emph{rank} is then the maximal cardinality of a
hyperedge. Then Sperner's Theorem combined with
\cref{lem:nosubsets}, also gives the following lower
bounds:
\begin{restatable}{corollary}{corsperner}\label{cor:sperner}
$\forall n, t$ and $k\geq{}2$, if $\ell_{k,n}=t$ then
$\binom{t}{\lfloor t/2\rfloor}\geq n$.
\end{restatable}

\begin{proof}
By Sperner's Theorem~\cite{Sperner:1928:clutter}, the size of any Sperner family with $t$ elements
is upper bounded by $\binom{t}{\lfloor t/2\rfloor}$.
Distributing keys for $t$ padlocks to $n$ participants while
satisfying \cref{lem:nosubsets} thus requires
$\binom{t}{\lfloor t/2\rfloor}\geq n$.
\end{proof}

\begin{restatable}{corollary}{corspernereven}\label{cor:spernereven}
$\forall n\geq 1$ and $\forall t\geq 2$ even, if $k\geq{}3$ and $\ell_{k,n}=t$ then
$\binom{t}{t/2}> n$.
\end{restatable}

\begin{proof}
By \cref{cor:sperner}, the only other possibility
is $n=\binom{t}{\lfloor t/2\rfloor}=\binom{t}{t/2}$.
But then the unique available Sperner family is that
of all subsets of equal size $t/2$.
In this family there exist pairs of subsets with an empty
intersection. The union of these two subsets is thus of exactly $t$
keys and must be able to open the door.
Therefore the threshold cannot be larger than~$2$.
\end{proof}

\begin{restatable}{lemma}{lemnthree}\label{lem:nthree}
$\ell_{2,3}\geq{3}$ and
$\forall n\geq 4, \ell_{2,n}\geq 4$.
\end{restatable}
\begin{proof}
Suppose $\ell_{2,n}=2$.
Then, if a single person has both keys, she can open both padlocks.
Hence, whatever the arrangement of padlocks, she can open the door alone and
$k<2$, a contradiction.
Therefore nobody can have more than one key.
As $k=2$, then two persons are sufficient to open the door.
They cannot have the same key by \cref{lem:nosubsets}.
But with only $2$ distinct keys and $n\geq 3$ people, at least two persons
must have the same key, a contradiction again.
Overall, $2$ padlocks are thus not enough.
For instance, $\ell_{2,3}\geq{3}$.
Finally, for $t=3$, \cref{cor:sperner}, shows that
$\binom{3}{1}=3\geq{n}$, thus $\forall n\geq 4, \ell_{2,n}\geq 4$.
\end{proof}

We have thus now for instance the following results:
\begin{itemize}
\item $\ell_{n,n}=n$ : use \cref{lem:ktwo} for the lower bound
  and our design for the upper bound.
\item $\ell_{2,3}=3$ and $\ell_{2,4}=4$ : use \cref{lem:nthree} for the lower bound
  and our design for the upper bound.
\end{itemize}

\subsection[Using log(n) padlocks for a threshold of 2 with n participants]{Using {${\mathcal O}(\log(n))$}
  Padlocks for a Threshold of {$2$} with
  {$n$} Participants}\label{sec:minimal2n}
Now we propose, in \cref{alg:twooutofn}, an arrangement
for a $2$-out-of-$n$ participants threshold system, using no more than
$n$ padlocks, and strictly less as soon as $n\geq 5$.
Indeed if the threshold is only $2$, then it is possible to reduce the
number of padlocks using our design. The idea is that whenever two
participants have a distinct set of keys then both of them have a
strictly larger set of keys than any of them taken separately.

\begin{algorithm}[htb]
\caption{Two-out-of-$n$ threshold system with shared keys}\label{alg:twooutofn}
\begin{algorithmic}[1]
\Require $n\geq{2}$, and $1\leq{i}\leq{t}\leq{n}$ such that $\binom{t}{i}\geq{n}$.
\Ensure A $2$-out-of-$n$ threshold padlock system with $t$ padlocks.
\If{$t<n$}
\State Set up an $(i+1)$-out-of-$t$ design of \cref{sec:our};
\State Create a total of $i{\cdot}n$ keys by copying the original $t$
keys, such that there are $n$ distinct subsets of $i$
keys;\hfill\Comment{Since $\binom{t}{i}\geq{n}$}
\State Give each participant a distinct $i$-tuple of keys.
\Else\Comment{If $t=n$, set $i=1$ and use directly our device of \cref{sec:our}}
\State Set up a $2$-out-of-$n$ design of \cref{sec:our};
\State Give each participant one of the $n$ keys.
\EndIf
\end{algorithmic}
\end{algorithm}

\begin{figure}[!ht]
  \begin{center}
\includegraphics[width=.45\textwidth]{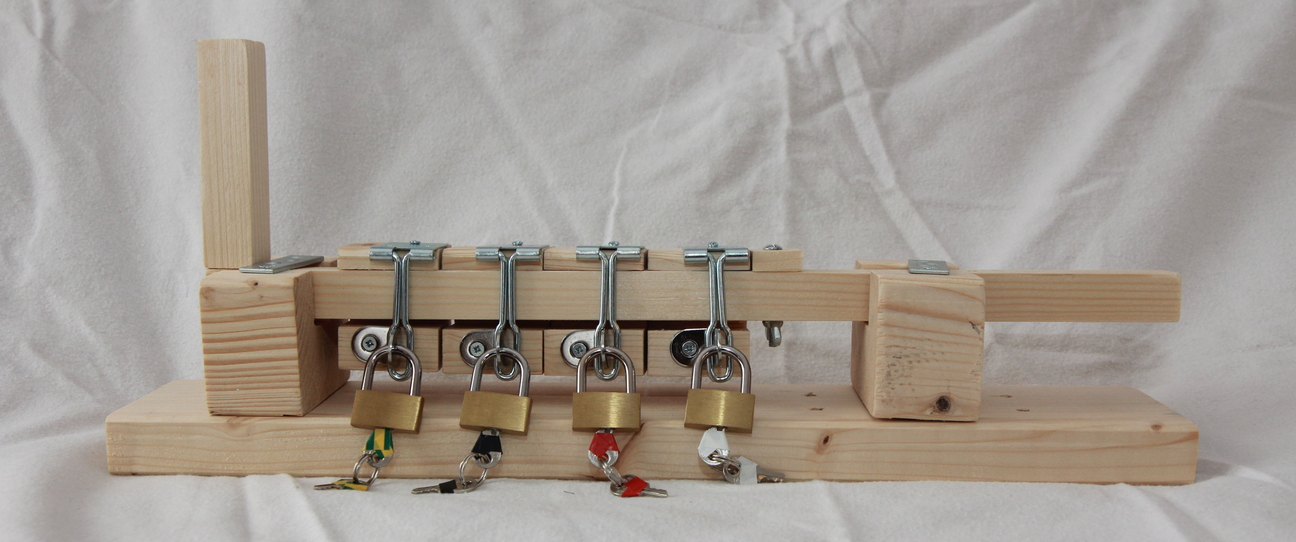}%
~~\includegraphics[width=.45\textwidth]{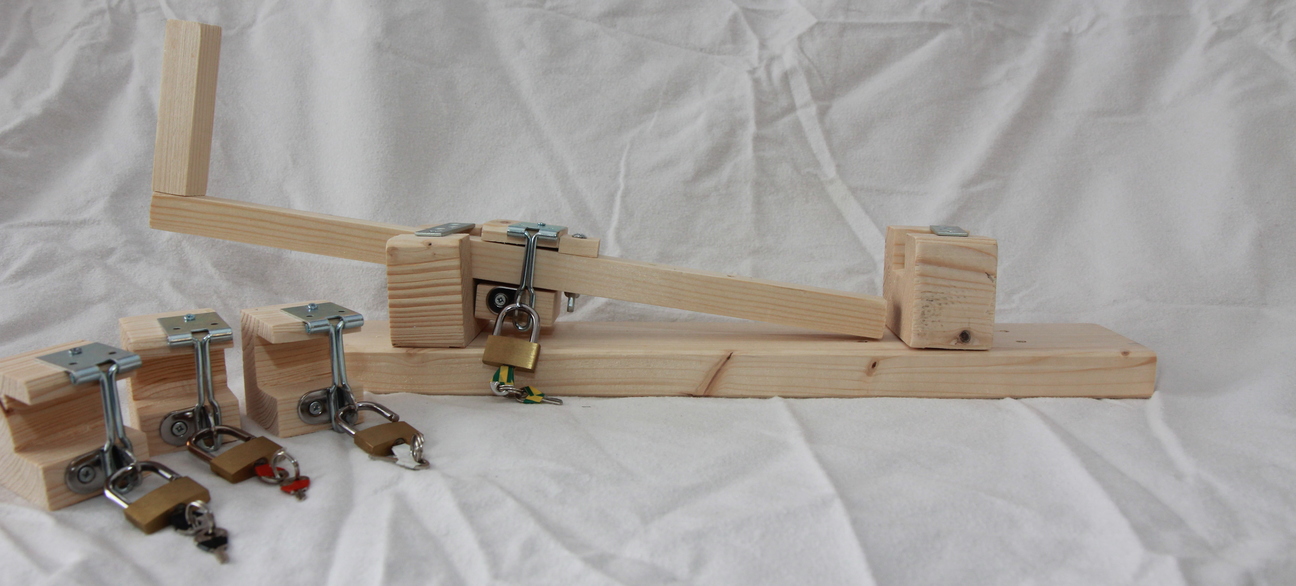}%
\vspace{0.2cm}
\includegraphics[width=.45\textwidth]{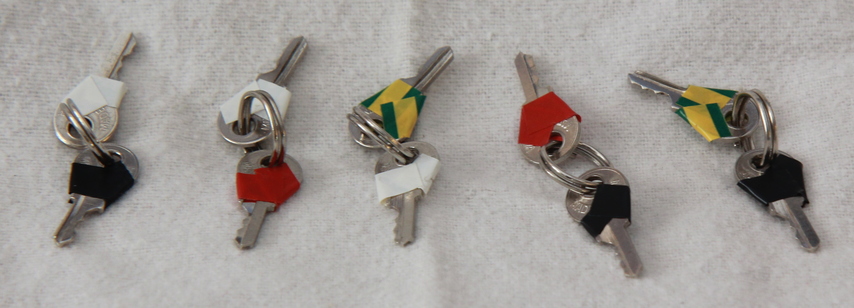}
\end{center}
\caption{Physical $3$-out-of-$4$ threshold padlock system. {\bf Left:}
  all locks closed. {\bf Right:} device opened using $3$ out of
  the $4$ locks. {\bf Bottom:} Example key distribution to
  achieve a $2$-out-of-$5$ threshold system using only $4$ padlocks
  and the $3$-out-of-$4$ device.
}\label{fig:lockthree}
\end{figure}

\cref{fig:lockthree} shows our prototype in a $3$-out-of-$4$
configuration.
Now, using \cref{alg:twooutofn}, this
configuration can also be used to implement a $2$-out-of-$5$ threshold
system with only
four padlocks by copying keys and distributing them in such a way that
each participant has a distinct subset of keys (as stated in
\cref{thm:five}). In this example, any two participants
together will have at least three different keys, which suffices to
open the $3$-out-of-$4$ device. This shows that $\ell_{2,5}=4$.

The correctness and the optimality of this schemme are proven by the
series of results in this section.

First, this scheme settles the small cases:
\begin{restatable}{theorem}{thmfive}\label{thm:five}
$\forall n\leq 5,  \forall k\geq 2, \ell_{k,n}=n$, except
$\ell_{2,5}=4$. We also have $\ell_{2,6}=4$ and $\forall n=7..10,\ell_{2,n}=5$.
\end{restatable}
\begin{proof}
  For any $n\leq 3$ and $2\leq k\leq n$ the results where already
  proven in \cref{lem:ktwo,lem:nthree}.  There remains
  $n=4$ and $n=5$, our design providing the upper bound. The proof is
  done by contradiction.

Let $t=\ell_{k,n}$ and suppose $t\leq n-1$.
Let $s$ be the number of participants having a single key.
These participants must have different single keys by
\cref{lem:nosubsets}.
The remaining $n-s$ participants must own at least $2$ keys, but
cannot own any of the first $s$ keys, by \cref{lem:nosubsets}.
Each set of keys of these remaining $n-s$ participants cannot be the
full set of the remaining $t-s$ keys, again by \cref{lem:nosubsets}.
Therefore, at least, the number $\delta_{s,t}$ of distinct subsets of size at least
$2$ and at most $t-s-1$ must be larger than $n-s$ (the requirement is
that the size of the clutter, must be larger than $n-s$, but in this
clutter all the subsets must at least be distinct).
This is:
\begin{equation}\label{eq:distinct}
\delta_{s,t}=\sum_{i=2}^{t-s-1} \binom{t-s}{i} \geq n-s
\end{equation}

But, if $t\leq n-1$ and $n\geq{4}$, then:
\begin{equation}\label{eq:ds}
\delta_{s,t}\leq\sum_{i=2}^{n-s-2}
\binom{n-s-1}{i}=2^{n-s-1}-1-(n-s-1)-(n-s-1)
\end{equation}

Therefore, \Cref{eq:distinct} cannot be satisfied whenever
\Cref{eq:ds} is $<n-s$, that is:
\begin{equation}\label{eq:smallcases}
2^{(n-s)-1} < 3(n-s)-1
\end{equation}

But \Cref{eq:smallcases} is true for $n-s \in
\{1,2,3,4\}$. Yet $n-s > 1$, otherwise $n-1\geq t\geq s$ implies at
most $t=s=n-1$, but then there remains no available key for the $n$-th
participant. Hence, we have $n-s \in \{2,3,4\}$.

For $n=4$, if $s\in\{2,1,0\}$ then \Cref{eq:smallcases} is
satisfied thus we can dismiss those cases. Finally, there remains no
value for $s$ meaning that our hypothesis $t\leq n-1$ is false.

For $n=5$, if $s\in\{3,2,1\}$ then \Cref{eq:smallcases} is
also satisfied so we can dismiss those cases. There remains the case
$s=0$ for $t=4$ (the case $t=3$ is excluded by the fact that
$\binom{3}{2}=3<5$).  The $5$ participants can thus only have $2$ or
$3$ keys each (if one of them has the $4$ keys he can open the door
alone).  If one of the $5$ participants owns $3$ keys $K_1,K_2,K_3$
then the other four must all own the fourth key $K_4$ (otherwise one
of them will own only a subset of the first $3$ keys, contradicting
\cref{lem:nosubsets}).  But then, excluding $K_4$, these four
remaining participants must have distinct non-included subsets of size
$1$ or $2$ of the $3$ keys $K_1,K_2,K_3$, which is impossible.
Therefore the rank of the arrangement is $2$, that is, all $5$
participants can only have $2$ keys each.  There are $\binom{4}{2}=6$
possible pairs. W.l.o.g. suppose that only the pair $K_3,K_4$ is not
among the participants pairs.

Then two participants owns $(K_1,K_3)$ for one and $(K_2,K_4)$ for the
other, so the two of them can open all the padlocks. This means
that~$k\leq{}2$. For $k>2$, we have a contradiction since no value $s$
can be taken, leading to refute the hypothesis $t\leq n-1$. Thus for
$n=5$ and $k\geq 3$ we have $\ell_{k,n}=n$.

The remaining case, $k=2$ is thus actually $2$-out-of-$5$ threshold
with at least $4$ padlocks where every player owns exactly $2$ keys.

This is satisfiable as follows: use a $3$-out-of-$4$ device with our
design with $4$ padlocks. Then provide the $5$ users with distinct
pairs of keys. Not a single user can open $3$ padlocks. But with
distinct pairs of keys all pairs of participants own at least $3$
different keys.

Finally, \cref{alg:twooutofn} gives a solution as soon as
$t$ is such that $\binom{t}{2} \geq n$,
while \cref{cor:sperner} prevents any solution
with $\binom{t}{\lfloor t/2\rfloor} < n$.
But with $t=4$ and
$t=5$, $\binom{t}{2}=\binom{t}{\lfloor{}t/2\rfloor}$.
So the
upper bound of \cref{alg:twooutofn} is also a lower bound.
Now $\binom{4}{2}=6$ and $\binom{5}{2}=10$ give the maximal
respective number of participants.
\end{proof}

Second, we give an asymptotic estimate for larger cases:
\cref{alg:twooutofn} makes it possible to implement a
$2$-out-of-$n$ threshold padlock system with only
$2\lceil\log_2(n)\rceil$ padlocks and $n\lceil\log_2(n)\rceil$ keys:
\begin{restatable}{proposition}{thmtwoi}\label{thm:twoi}
\cref{alg:twooutofn} correctly provides a $2$-out-of-$n$
threshold padlock system and for $n\geq{2}$,
$\ell_{2,n}\leq{2\lceil\log_2(n)\rceil}$.
\end{restatable}

\begin{proof}
Consider an $(i+1)$-out-of-$t$ threshold system with $t$ padlocks
for $\binom{t}{i}\geq{n}$.
Distribute $i$ keys for each participant, such that all the hyperedges
are distinct. This is possible as $\binom{t}{i}\geq{n}$.
No single participant can open the device, but any two participants
have different hyperedges of size $i$ and thus have at least $i+1$
distinct keys. This is enough to open the door
and~\cref{alg:twooutofn} is correct.
Finally,
$\binom{t}{i}\geq{(t/i)^i}$, so we can for instance
set $i=\lceil\log_2(n)\rceil$, so that each participant gets
that many keys, and setup $t=2\lceil\log_2(n)\rceil$ padlocks,
as $2^{\lceil\log_2(n)\rceil}\geq{n}$.
\end{proof}

For instance, the first case where triples are better than couples in
\cref{alg:twooutofn} is for $n=16$.
As $\binom{6}{2}=15$ and $\binom{7}{2}=21$, with pairs the Algorithm
would use $7$ padlocks, where $6$ are enough:
setup a $4$-out-of-$6$ device and give distinct triples of copies of
the $6$ keys to each of the $16$ participants.
This is possible as $\binom{6}{3}=20\geq{16}$. Then any pair of
participants have at least $3+1=4$ different keys and they can open
our device.
Overall, we have that the minimal number of padlocks for a
$2$-out-of-$n$ threshold system is ${\mathcal O}(\log(n))$ with
${\mathcal O}(n\log(n))$ keys. Indeed,
the lower bound is given in~\cref{cor:sperner}, and it is realizable by
\cref{alg:twooutofn}. We thus have proven~\cref{thm:twoopt}.

\begin{restatable}{corollary}{thmtwoopt}\label{thm:twoopt}
For $n\geq{2}$, $\ell_{2,n}=\min\left\{t,~\text{s.t.}~\binom{t}{\lfloor{t/2}\rfloor}\geq{n}\right\}$.
\end{restatable}

\begin{proof}
The lower bound is given by \cref{cor:sperner}.
For the upper bound consider as in \cref{thm:twoi}
an $(\lfloor{t/2}\rfloor+1)$-out-of-$t$ threshold system with $t$
padlocks and distribute $\lfloor{t/2}\rfloor$ keys to each participant,
such that all the hyperedges are distinct.
\end{proof}

\input{directlog}

%% file: directlog.tex
  \subsection[A Trick for 2-out-of-n Padlock System with
    Exactly 2log(n) Padlocks]{A Trick for $2$-out-of-$n$ Padlock System with
    Exactly $2\lceil\log_2(n)\rceil$ Padlocks}

  It is possible to directly
  obtain a $2$-out-of-$n$ padlock system with exactly
  $2\lceil\log_2(n)\rceil$ padlocks and no external device. To see
  this, one can mix $2$-out-of-$2$ devices and a
  $1$-out-of-$\log_2{n}$.  Then, remark first that daisy chains are
  $1$-out-of-$t$ devices and second that setting $t$ padlocks on the
  same latch provides a $t$-out-of-$t$ device. For the illustration
  purpose, we describe the alternative construction in two steps:
  first with devices, then without any device.

The first construction is as follows:
\begin{enumerate}
\item Consider $2$-out-of-$2$ devices, similar to those
  of~\cref{fig:lock6o6}, and take $\lceil\log_2{n}\rceil$ of them;
  On each one of these $2$-out-of-$2$ devices set two padlocks, one
  black and one white. This is $2\lceil\log_2{n}\rceil$ distinct padlocks;
\item Attach these $\lceil\log_2{n}\rceil$ devices to each latch of a
  $1$-out-of-$\lceil\log_2{n}\rceil$ device;
\item Order the participants and give them a distinct number between
  $0$ and $2^{\lceil\log_2{n}\rceil}$. Then give each one of the
  participants $\lceil\log_2{n}\rceil$ keys, following the binary
  digits of her number.
  On the one hand, if the $i$-th bit of her number is zero, then give
  her the key of the white padlock of the $i$-th $2$-out-of-$2$
  device;
  on the other hand, if the $i$-th bit of her number is one, then give
  her the key of the black padlock of the $i$-th $2$-out-of-$2$
  device.
\end{enumerate}
This is a total of $2n\lceil\log_2{n}\rceil$ keys.
Now each participant alone cannot open any latch (she owns only one of
the two keys required for that), therefore she cannot open the door.
Differently, any two participants have at least one bit, $i$, of
difference. For this bit, the two of them thus have both keys of the
$i$-th $2$-out-of-$2$ device. They can thus open it, thus open the
$1$-out-of-$\lceil\log_2{n}\rceil$ device and open the door. This is
overall a $2$-out-of-$n$ system.

Now the second construction mimics the first one given above, but without any
particular device. It is shown in~\cref{alg:doubledaisy}.
\begin{algorithm}[htb]
\caption{Two-out-of-$n$ double daisy chain with only
  $2\lceil\log_2{n}\rceil$ padlocks}\label{alg:doubledaisy}
\begin{algorithmic}[1]
\State Setup a daisy chain with $\lceil\log_2{n}\rceil$ white padlocks;
\State For each white padlock, double it with a black padlock: that is
  form a daisy chain with double links;
\State Similarly, each participant receives $\lceil\log_2{n}\rceil$
  keys, white or black, according to the binary decomposition of her
  number.
\end{algorithmic}
\end{algorithm}

First, similarly, each participant alone cannot open the chain, as she cannot
open any link, having only one of the two keys required to open one
link. Second, similarly also, any two participants having different numbers
have at least both keys of one double link and can open the chain and
the door.

This is simple and does not require any additional device, apart from
$2\lceil\log_2{n}\rceil$ padlocks. Note that this usually uses more
padlocks than~\cref{alg:twooutofn}. For instance, for a $2$-out-of-$8$
system, the above method requires $2*3=6$ padlocks and for a
$2$-out-of-$10$ system, it requires $2*4=8$ padlocks,
while~\cref{thm:twoopt} uses only $5$ padlocks for both cases as
$\binom{5}{2}=10$.

%% file: algebra.tex
\section{Access Structures}\label{sec:algebra}
In a secret sharing scheme, a datum $d$ is broken into shadows which are
shared by a set of trustees. The family $\{G\subseteq{P}:G~\text{can
  reconstruct}~d\}$ is called the \emph{access structure} of the scheme. A
$k$-out-of-$n$ scheme is a secret sharing scheme having the access
structure $\{G\subseteq{P}:|G|=k\}$~\cite{Ito:1993:accessstruct}. 
In this section we show how to physically implement access structures
defined by logic gates. 
Numeric solutions with interpolation usually use one evaluation point
for each literal and one polynomial per clause.
Our physical solution uses instead only one padlock for each distinct
variable and one device per whole normal form.

\subsection{Towards a Padlock Algebra with One Device per Normal
  Form}\label{ssec:normalforms}
A~generalization of threshold schemes is to be able to implement any
access scheme described by a logic formula.
This is possible by implementing AND and OR gates, as shown in
\cref{prop:cnf}
and \cref{alg:disjunctive,alg:conjunctive},
following~\cite{Benaloh:1990:crypto88,Ito:1993:accessstruct}.
A first idea is to use chains so that opening a padlock actually frees
a chain that can free several latches. Then a second idea is that
$1$-out-of-$n$ systems are just like a disjunction while
$n$-out-of-$n$ systems are just like a conjunction.

\cref{alg:disjunctive} shows how to generate a padlock system openable
by any satisfiable realization of a disjunction with $t$ clauses and
$n$ distinct variables. For this, a single $1$-out-of-$t$ device is set. It will
open if any of the $t$ conjunctive clauses is true.
Associate each latch of the device to one conjunctive clause.
Then associate one padlock for each variable. To simulate the
subjection of a clause to a variable, each padlock closes a chain
passing through  each latch corresponding to a clause containing that
variable (and thus preventing the opening of those latches if that
padlock is not open).

\begin{algorithm}[htbp]
\caption{Physical DNF  with one padlock for each variable}\label{alg:disjunctive}
\begin{algorithmic}[1]
\Require A disjunctive normal form with $t$ clauses and $n$ variables.
\Ensure A system with $n$ padlocks, openable by any satisfiable realization of
the normal form.
\State Set up a $1$-out-of-$t$ threshold system, with one
latch for each clause;
\State For each variable present in the formula: pass a chain through
the hole of each latch corresponding to a conjunction containing that
variable; close that chain with one padlock.
\end{algorithmic}
\end{algorithm}

\cref{fig:dnf} gives an example of \cref{alg:disjunctive} on the
logic formula $(A\wedge{B})\vee(A\wedge{C})\vee(B\wedge{D})\vee(E)$.
\begin{figure}[htbp]\centering
\includegraphics[width=.4\textwidth]{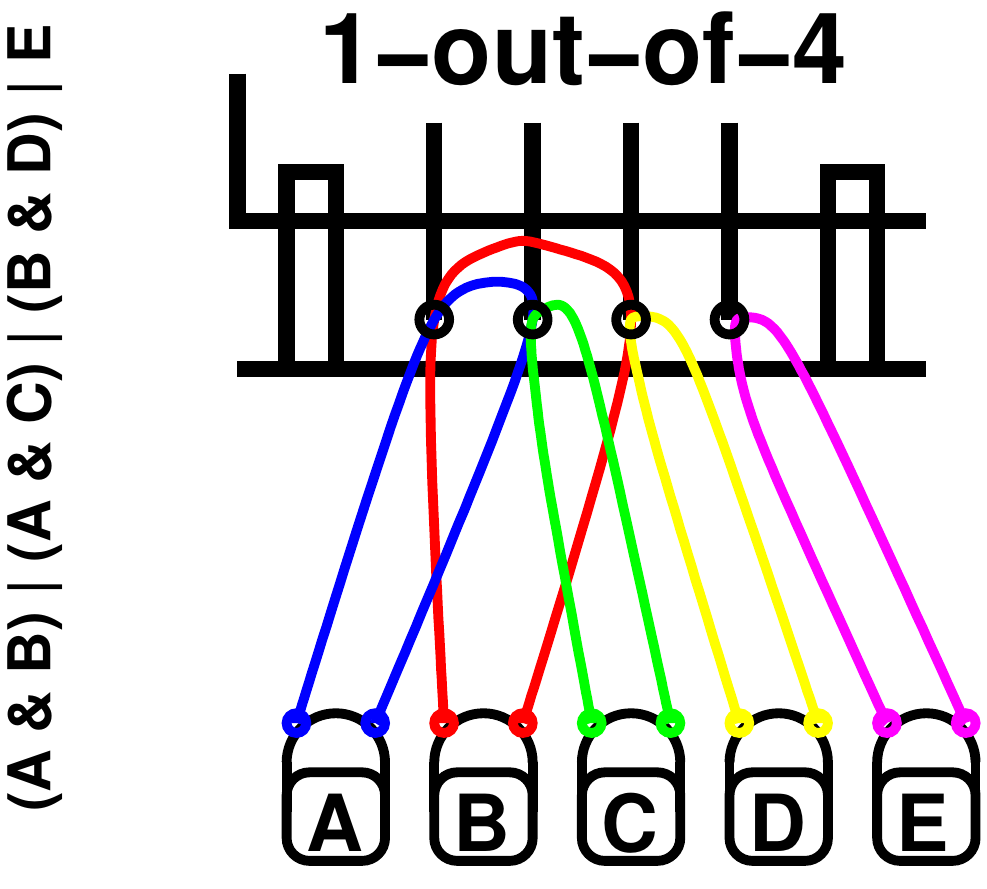}
\caption{\cref{alg:disjunctive} on a disjunctive normal form: one
  padlock closing one chain per term, a $1$-out-of-$4$ device using a
  single latch for each clause.}\label{fig:dnf}
\end{figure}

Now for conjunctions, we instead use a $t$-out-of-$t$ master structure
and several other $1$-out-of-$k$ systems, one for each conjunction in
the CNF, as shown in~\cref{alg:conjunctive}.
\begin{algorithm}[htbp]
\caption{Physical CNF with one padlock for each variable}\label{alg:conjunctive}
\begin{algorithmic}[1]
\Require A conjunctive normal form with $t$ clauses and $n$ variables.
\Ensure A system with $n$ padlocks, openable by any satisfiable realization of
the normal form.
\State Set up a $t$-out-of-$t$ threshold system;
\State Set up one $1$-out-of-$k_i$ threshold system for each
conjunctive clause with $k_i$ variables. Attach each $1$-out-of-$k_i$
system to one of the $t$ latches of the $t$-out-of-$t$ system.
\State For each variable present in the formula: pass a chain through
  a free latch of each $1$-out-of-$k_i$ system corresponding to a
  clause containing that variable; close that chain with one padlock.
\end{algorithmic}
\end{algorithm}

\cref{fig:cnf} gives an example of \cref{alg:conjunctive} on the
logic formula $(A\vee{B}\vee{C})\wedge(D)\wedge(C\vee{E})$.
\begin{figure}[htb]\centering
\includegraphics[width=.5\textwidth]{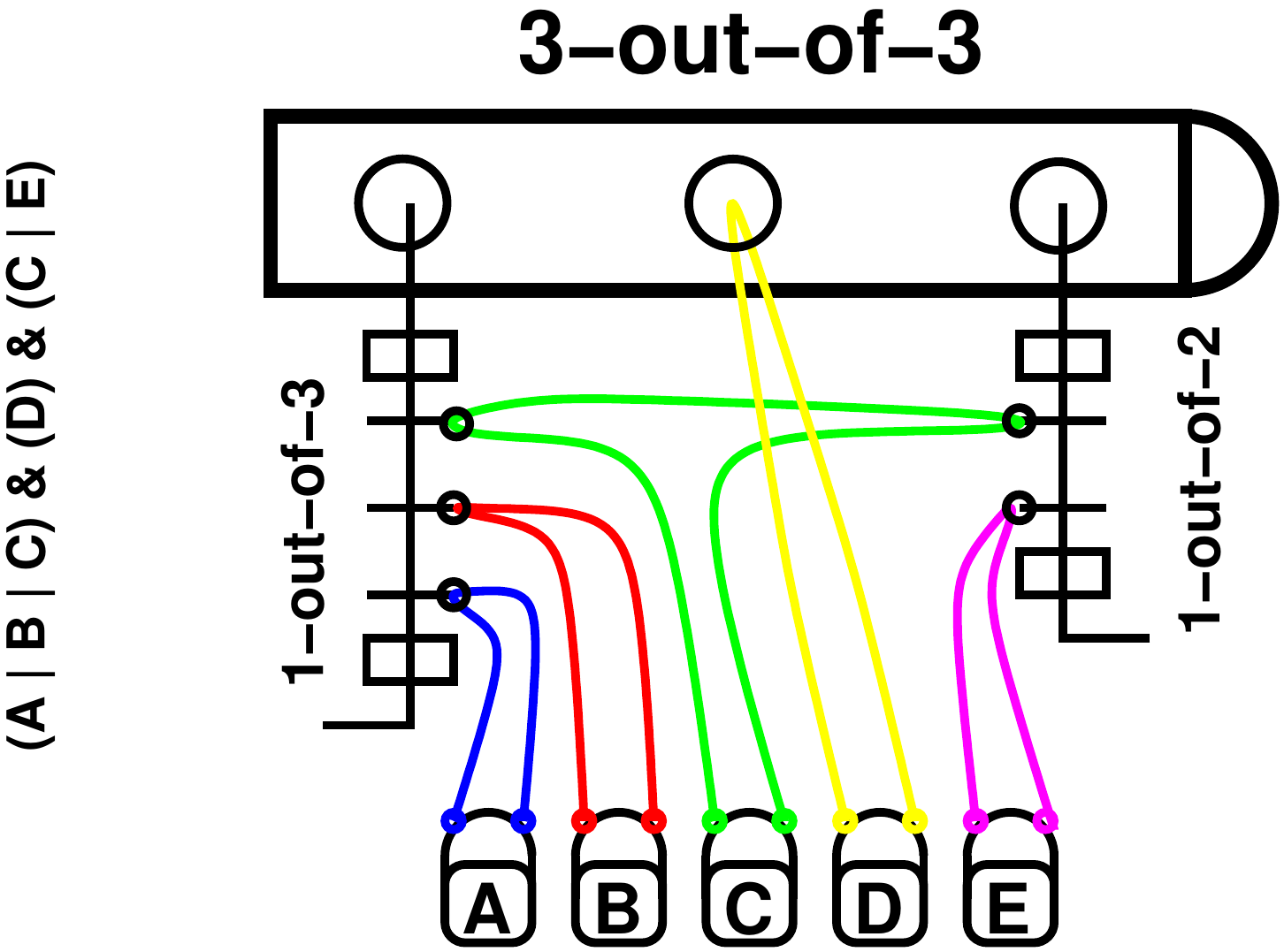}
\caption{\cref{alg:conjunctive} on a conjunctive normal form: one
  padlock closing one chain per term, a $4$-out-of-$4$ device using a
  single $1$-out-of-$k_i$ device for each clause (if $k_i>1$).}\label{fig:cnf}
\end{figure}

Both \cref{alg:disjunctive,alg:conjunctive} thus provide a way to
build systems with a number of padlocks equal to the number of
distinct variables in the normal form: this is \cref{prop:cnf}
thereafter.

\begin{proposition}\label{prop:cnf}
Any disjunctive or conjunctive normal form with $t$ clauses, $m$
distinct variables and no negation is realizable with $m$ padlocks
\end{proposition}
\begin{proof}
First for disjunctive clauses: they require one $1$-out-of-$t$
threshold system and $m$ chains, as shown in
\cref{alg:disjunctive}. The ``door'' can be opened only by a
satisfiable interpretation where TRUE means opening the padlock and
FALSE means letting it closed.

Similarly one can create arrangements for conjunctive normal forms,
also with as many padlocks as there are distinct variables as shown in
\cref{alg:conjunctive}.
\end{proof}

\input{algebra_blocks}

\input{knotted}

%% file: algebra_blocks.tex
\subsection{Further examples of logic formulae}\label{app:blocks}

In \cref{sec:algebra}, we show that any access scheme described
by a logic formula without negation can be implemented using simple
physical devices.
In this section, we show some other constructions that can simplify
the use of \cref{alg:disjunctive,alg:conjunctive} for normal forms.
We also show how our physical methods can implement some formulae
that are proven impossible with a single secret sharing scheme.

First, to implement \cref{alg:conjunctive} we need a
$1$-out-of-$k_i$ system for each clause. If this is simpler, one can
always build such a system by composing small $1$-out-of-$2$ systems.
For instance \cref{fig:etouet}, left, shows how to create one tree for
each disjunctive clause as used in \cref{alg:conjunctive}:
assemble U-shaped metal rods. It is also possible to create a daisy
chain of $1$-out-of-$2$ devices like the one in
\cref{fig:etouet}, right.
\begin{figure}[htbp]
\hfill \includegraphics[height=2.5cm]{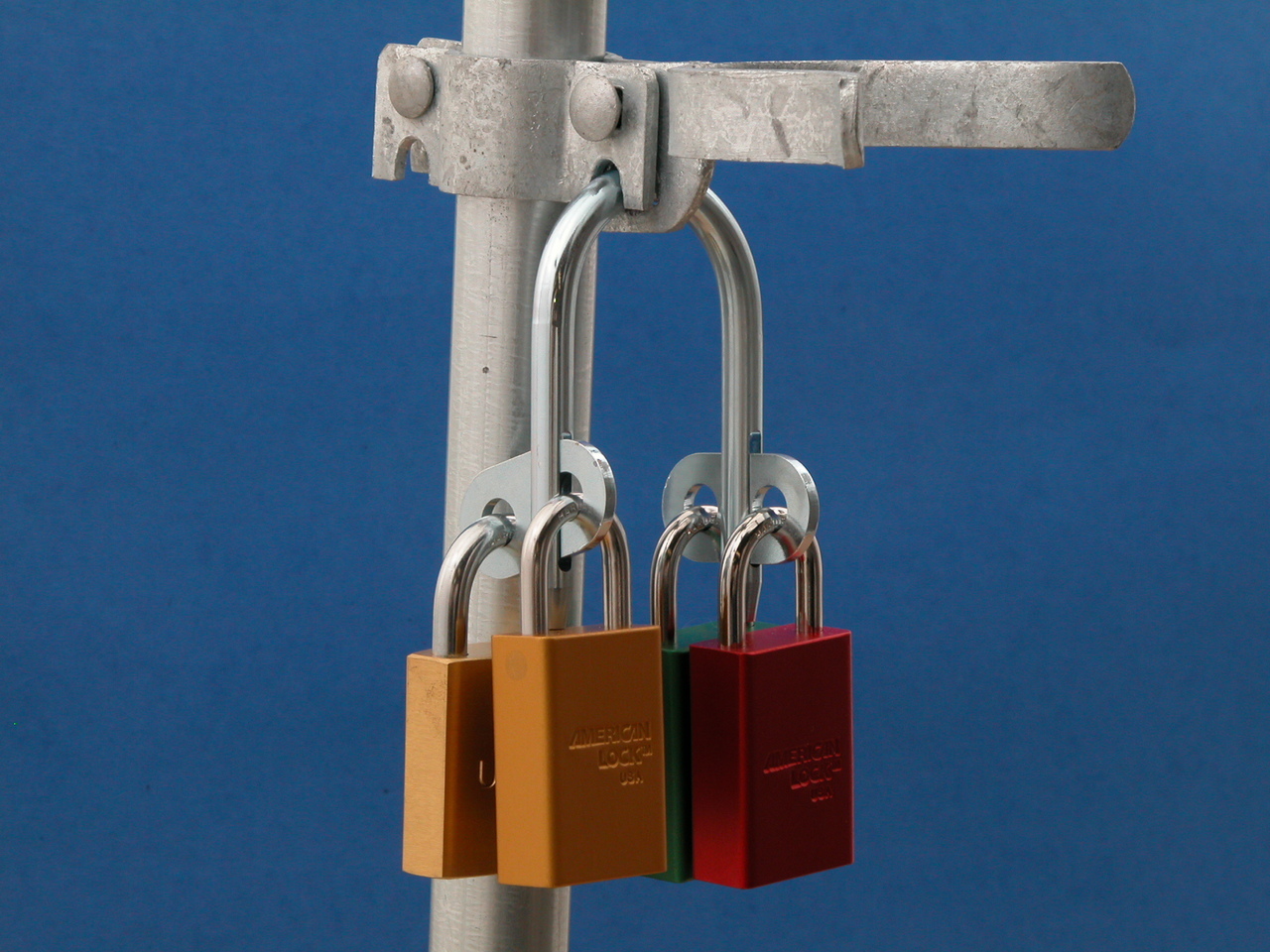}\hfill
\includegraphics[height=2.5cm]{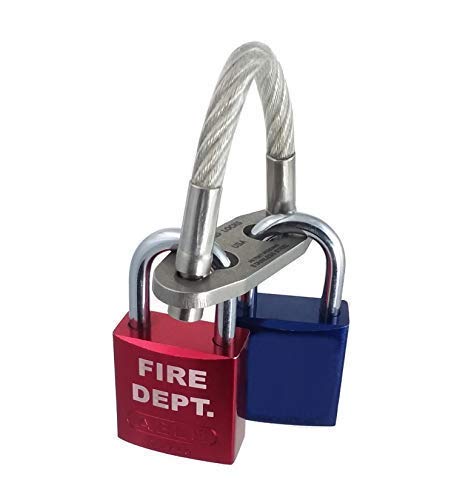}\hfill~
\caption{Left, a tree-like disjunctive clauses system by
  GateKeeper GM P6006 to combine 2, 3 or 4 locks; right, a physical
  $1$-out-of-$2$ disjunction Model Cb2 by Sharelox.}\label{fig:etouet}
\end{figure}

Now, second, we give examples of usage of all our devices and construction.
For this we offer physical solutions for two examples, proven
unrealisable using a single scheme. Indeed,
\cite{Benaloh:1990:crypto88} shows that the two following cases cannot
be solved if users must use the same system of shares:
\begin{enumerate}
\item $(A\wedge{}B)\vee(C\wedge{}D)$
\item $(A\wedge{}B)\vee(C\wedge{}D)\vee(B\wedge{}C)$
\end{enumerate}

%

%

%


However, with a physical system, we can implement such access
schemes with somewhat less devices as we have physical tools to
combine conjunctions and disjunctions:
\begin{itemize}
\item Conjunctions can be implemented with $n$-out-of-$n$ systems as in
\cref{fig:lock6o6};
\item Disjunctions can be implemented with $1$-out-of-$n$ systems as in
\cref{fig:lock1o6,fig:lock1o10T} or in \cref{sec:our}.
\end{itemize}

For the first formula:
$(A\wedge{}B)\vee(C\wedge{}D)$, we report no improvements.
One can implement \cref{alg:disjunctive} on this
formula, and provide a physical solution:
use a daisy chain of two $2$-out-of-$2$ classical
equivalent system, as in \cref{fig:lock6o6}, one for each
conjunction. This is not different from the solution of
\cite{Benaloh:1990:crypto88} with two distinct secret sharing schemes.

Now, for the second formula,
$(A\wedge{}B)\vee(C\wedge{}D)\vee(B\wedge{}C)$, a naive
implementation (resp. \cite{Benaloh:1990:crypto88} solution) would
require six padlocks (resp. $6$ shares for three systems of $2$
shares, one for each clause). But it is possible to use only four
padlocks, as shown in \cref{alg:Benaloh} and \cref{fig:Benaloh}.

\begin{algorithm}[htbp]
\caption{Physical realization with only $4$ padlocks
  of~\cite[Theorem~3]{Benaloh:1990:crypto88}}\label{alg:Benaloh}
\begin{algorithmic}[1]
\State Setup a~$1$-out-of-$3$ threshold system (thus with~$3$ latches);
\State Put a padlock~$A$ on the first latch and a padlock~$D$ on the
third latch;
\State Pass a chain in the holes of the first two latches and close
that chain with a padlock~$B$;
\State Pass a chain in the holes of the last two latches and close
that chain with a padlock~$C$.
\end{algorithmic}
\end{algorithm}

\begin{figure}[htbp]\centering
\includegraphics[width=.35\textwidth]{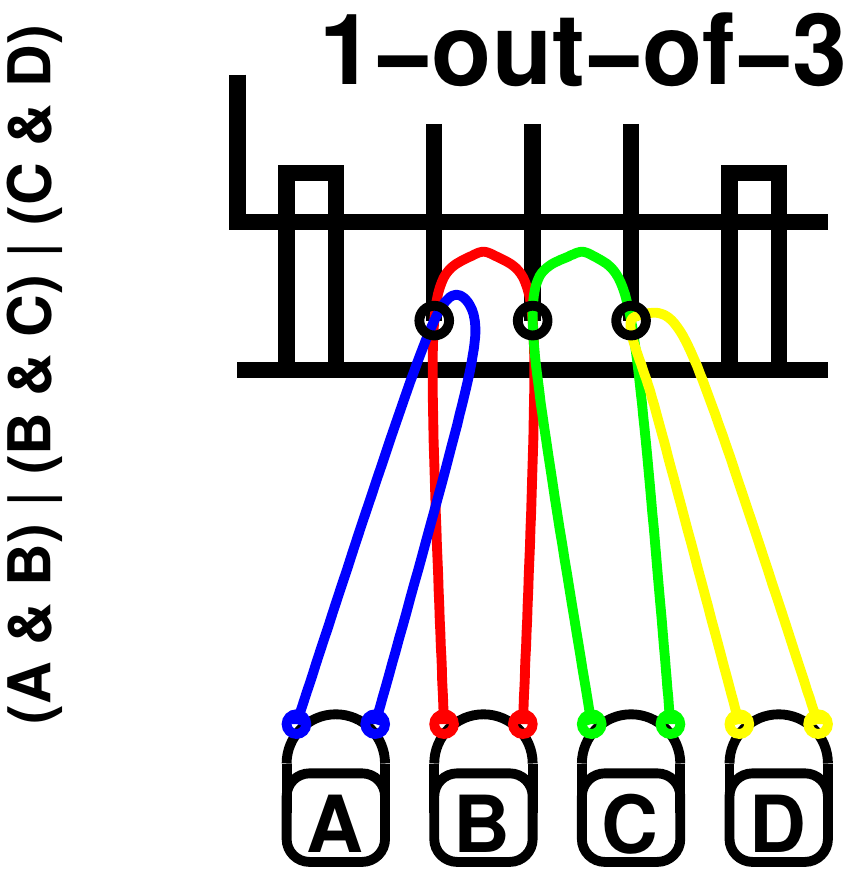}
\caption{\cref{alg:Benaloh} on
  $(A\wedge{}B)\vee(C\wedge{}D)\vee(B\wedge{}C)$ with one padlock per
  variable (and not per literal in the formula).}\label{fig:Benaloh}
\end{figure}

\begin{remark}
Finally, note that, with our novel design, \cref{prop:cnf} is not
optimal.
Consider for instance the DNF with a single participant able to open
the door or any two among five others:
$A
\vee
(B\wedge{}C)
\vee
(B\wedge{}D)
\vee
(B\wedge{}E)
\vee
(B\wedge{}F)
\vee
(C\wedge{}D)
\vee
(C\wedge{}E)
\vee
(C\wedge{}F)
\vee
(D\wedge{}E)
\vee
(D\wedge{}F)
\vee
(E\wedge{}F)
$. \cref{prop:cnf} would require~$6$ padlocks and a
$1$-out-of-$11$ design.
However, we can use \cref{thm:five} and our design for
a~$3$-out-of-$4$ lock with only~$4$ padlocks as described %
thereafter: set up a~$3$-out-of-$4$ design with~$4$ padlocks;
give pairs of distinct keys to each participant~$B,C,D,E,F$;
give any three distinct keys to~$A$.
\end{remark}

%% file: knotted.tex
\newcommand{\anneau}{\ensuremath{\circledcirc}}
\subsection{Knotted padlocks}\label{ssec:knotted}
A post on
\href{https://crypto.stackexchange.com}{crypto.stackexchange.com} by
Ahle~\cite{Ahle:2012:knots} hints that one could create~$k$-out-of-$n$
threshold padlock systems using~$n$ padlocks and some wire.

His idea is to have the wire securing the door and going around the
rings of the padlocks in a certain configuration. If a padlock is
opened then it frees his part of the wire  and potentially more from
other padlocks. The example given is for a~$1$-out-of-$2$ system:
``\emph{\small Say you have one wire to which the [door] is fastened
  and two padlocks. You want that if either of the locks are opened,
  the wire is completely freed. You do this by letting the
  wire go first clockwise around [the ring of] the first [padlock],
  then clockwise around the second, then anticlockwise around the
  first and finally anticlockwise around the second. It can be thought
  of as~$aba^{-1}b^{-1}$. If you remove either, the other cancels
  out. It generalizes to any k out of n padlocks}''.

This is a neat idea, which however turns out to not generalize easily to any~$k$ out of~$n$ system, though.

First, associate a variable from a non-commutative group to each one
of $n$ padlocks. To simulate the opening of a padlock set this variable
to~$1$, the neutral element of the group, seen multiplicatively.
Suppose that this variable represents one clockwise wrapping of the
wire and that the inverse of that variable represents the anticlockwise wrapping.
Then, the sequential arrangement of the wrappings of the wire around
the rings of the padlocks is a sequential multiplication of these
variables and their inverses, just like a Knot group presentation.

For instance, if a clockwise wrapping is directly followed by an
anticlockwise wrapping then this is useless and represented by
$x x^{-1}=1= x^{-1} x$. Finally, if the door is opened when some padlocks are
opened then it means that the multiplication of the variables is equal
to~$1$ when the variables associated to the opened padlocks are set
to~$1$. We will say that padlocks are \emph{knotted} if there is a
sequential wrapping of a wire around the rings of its padlocks.

On the one hand, we see that the equation $aba^{-1}b^{-1}$ represents a
generic OR gate: for the two padlock case, if one variable is set
to~$1$, then either $aba^{-1}b^{-1}=bb^{-1}=1$ or
$aba^{-1}b^{-1}=aa^{-1}=1$. This also generalizes to creating the OR of
any \emph{independent} subsystems: if~$X$ and~$Y$ are two equations for two knotted
systems with \emph{distinct padlock sets}, then~$XYX^{-1}Y^{-1}$ is the
equation of the OR of these two systems.

If the equations are not independent, then some cancellations can
occur. Consider for instance the formula~$a\vee{a}$; then
$aaa^{-1}a^{-1}$ is always~$1$ even if~$a$ is not open.
Now, in order to prevent such cancellations, it is possible to
surround a set of equations by an independent padlock and its inverse. Then no
cancellation can happen. Even better, one can use a simple ring (this
is a padlock that nobody can open) and wrap around it one way before
the equation and the other way after the equation. For that simple
ring, denote by \anneau{} one clockwise wrapping of the wire
around it (and by~$\anneau^{-1}$ an anticlockwise
wrapping). Note that the latch of the door, if any, could be used as this ring too.
In any case, for instance,
$X{\anneau}Y{\anneau}^{-1}X^{-1}{\anneau}Y^{-1}{\anneau}^{-1}$ then represents a
generic OR gate where the subsystems need not be independent. Indeed, if,
and only if, any of~$X$ or~$Y$ is~$1$ then everything collapses.

On the other hand, to represent an AND gate between two independent padlock
systems, then simply multiplying both equations suffices, in any order
and with any inverse (i.e., independence ensures that~$XY$,~$YX$,
$X^{-1}Y^{-1}$, $Y^{-1}X^{-1}$, $XY^{-1}$, etc. all represent the
conjunction). Similarly, one can enclose dependent subsets of padlocks
with the simple ring \anneau.

For instance some access structures of the previous subsections can also
be realized this way with one padlock per literal:
\begin{itemize}
\item $(A\wedge{}B)\vee(C\wedge{}D)$ can be represented by $X=abcdb^{-1}a^{-1}d^{-1}c^{-1}$;
\item $(A\wedge{}B)\vee(C\wedge{}D)\vee(B\wedge{}C)$ by
  $X{\anneau}bc{\anneau^{-1}}X^{-1}{\anneau}c^{-1}d^{-1}\anneau^{-1}$.
\end{itemize}

There is a nice linear setup, for~$(n{-}1)$-out-of-$n$ threshold
systems, as shown in \cref{lem:nmoinsun}.
\begin{lemma}\label{lem:nmoinsun}
The knotted padlock system, setup with with~$n$ padlocks and a
wire, and wrapped~$2n$ times, following the presentation
$x_1x_2\ldots{x_n}x_1^{-1}x_2^{-1}\ldots{x_n^{-1}}$, is a
$(n{-}1)$-out-of-$n$ threshold padlock system.
\end{lemma}
\begin{proof}
Set any subset of size~$n-1$ of the variables to~$1$, there remains
$x_jx_j^{-1}$.
For any subset of size~$n-2$ or less, there would remain at least
$x_ix_jx_i^{-1}x_j^{-1}$ with~$i\neq{j}$ and the door is not freed.
\end{proof}

This setup actually is optimal as shown in \cref{lem:knottedlower}

\begin{lemma}\label{lem:knottedlower} Let $k\geq{1}$ and $n\geq{k+1}$,
A knotted padlock $k$-out-of-$n$ threshold system, setup with with~$1$
padlock per participant and a wire, requires an even number of
wrappings, and at least $2n$ of them.
\end{lemma}
\begin{proof}
Let $a_i$ be the number of wrappings around padlock $i$.
Suppose that $k$ participants not including $i$ open their
padlock.
This is possible since $n>k$.
Then the system must be freed. Therefore, any clockwise
wrapping around $i$ must be accompanied by an anticlockwise one. This
shows that $a_i=2\alpha_i$ and that the total number of wrappings is
even.
Now suppose that there are no wrappings around padlock $i$. Then the
participant $i$ is useless in opening the system. This contradicts
the notion of a threshold system. Finally, we have that $a_i\neq{0}$
and thus that $a_i\geq{2}$. The total number of wrappings is thus
larger than $2n$.
\end{proof}

By
\cref{lem:nmoinsun,lem:knottedlower} we have an optimal linear
knotted system for $(n{-}1)$-out-of-$n$ threshold padlock systems.
But, unfortunately, we have no simple candidate for other thresholds.
Generic threshold system can be implemented with
simple gates, but then they must use an exponential number of
them~\cite{Smolensky:1987:lowcircuits}. So this method of knotting the
padlocks might not be directly practical.
For instance, an exhaustive search of the $\sum_{j=3}^4
6^{2j}=1\,726\,272$ formulas with $3$ variables and their $3$ inverses
(since the number of terms must be even and larger than $2*3$ by
\cref{lem:knottedlower}), showed that no formula exists for a
$1$-out-of-$3$ threshold system with strictly less than $10$
terms. The smallest one is thus a permutation of
$abcb^{-1}c^{-1}a^{-1}cbc^{-1}b^{-1}$, with $10$ wrappings.

For more generic thresholds we were only able to devise a solution with an exponential
number of wrappings, as shown in \cref{alg:knotted}, again loosing
practicality for most of these knotted systems.

\begin{algorithm}[htb]
\caption{Knotted padlock threshold system}\label{alg:knotted}
\begin{algorithmic}[1]
\Require $k\geq{1}$, a wire and $n\geq{k}$ padlocks.
\Ensure A $k$-out-of-$n$ threshold padlock system with $n$ knotted
padlocks.
\If{$k==n$}
\State\label{lin:noon}\Return the wiring of all the padlocks together, with
presentation $x_1x_2\ldots{x_{n-1}}\ldots{x_n}$.
\EndIf
\If{$k==n-1$}
\State\label{lin:nmuoon}\Return the wiring of all the padlocks together, with
presentation $x_1x_2\ldots{x_n}x_1^{-1}x_2^{-1}\ldots{x_n^{-1}}$.
\EndIf
\If{$k==1$}
\State Recursively compute a presentation $X$ for a
$1$-out-of-$(n{-}1)$ system with the last $n-1$ padlocks;
\State\label{lin:1oon}\Return the wiring of all the padlocks together, with
presentation $x_1Xx_1^{-1}X^{-1}$.%
\EndIf
\State Recursively compute a presentation $X$ for a
$(k{-}1)$-out-of-$(n{-}1)$ system with the last $n-1$ padlocks;
\State Recursively compute a presentation $Y$ for a
$k$-out-of-$(n{-}1)$ system with the last $n-1$ padlocks;
\State\Return the wiring of all the padlocks together, with
presentation $x_1X{\anneau}Y{\anneau}^{-1}X^{-1}x_1^{-1}{\anneau}Y^{-1}{\anneau}^{-1}$.
\end{algorithmic}
\end{algorithm}

\begin{theorem} Let $k\geq{1}$ and $n\geq{2}$,
\cref{alg:knotted} is correct and requires a number of wrappings
$W(k,n)$ that satisfies:
\begin{itemize}
\item $W(1,n)=\frac{3}{2}2^n-2$;
\item $W(n-1,n)=2n$;
\item $W(n,n)=n$;
\item For $k\in[2..(n-2)]$, $W(k,n)\geq{\frac{3}{2}2^{n}+6}$.
\end{itemize}
\end{theorem}
\begin{proof}
For the correctness, we look at the cases.
If $k=n$ then this is a $n$-out-of-$n$ system.
All padlocks are wired, one after the other a single time. Therefore
no simplification can occur by opening padlocks. This means that all
padlocks must be opened to free the system and that no strict subset of
owners can open the system.

The $k=n-1$ case is settled by \cref{lem:nmoinsun}.

If $k=1$. Then we proceed by induction. We have seen that a
$1$-out-of-$2$ system is indeed represented by a formula
$aba^{-1}b^{-1}$. Now suppose that we have a presentation $X$ valid
for a $1$-out-of-$(n-1)$ system with $n-1$ padlocks. Then an additional
participant, numbered $1$, uses a new padlock and the overall
presentation is $E=x_1Xx_1^{-1}X^{-1}$. If $x_1$ is opened then
$E=XX^{-1}$ cancels out. If any of $x_2,\ldots,x_{n}$ is opened then
$X$ cancels and $E=x_1x_1^{-1}$ also cancels out. Therefore any
participant alone can open the system. Conversely, the system does not
collapse: by induction, first, neither $X$ nor $X^{-1}$ is $1$ if no
padlock is opened. Second $x_1$ and $X$ are using independent sets of
padlocks so no cancellation can occur between $x_1X$, $Xx_1^{-1}$, nor
$x_1^{-1}X^{-1}$.

Finally the generic case, with now
$E=x_1X{\anneau}Y{\anneau}^{-1}X^{-1}x_1^{-1}{\anneau}Y^{-1}{\anneau}^{-1}$,
is also handled by induction.
$k$ participants are either $x_1$ and $k-1$ others or $k$ other than
$x_1$.
In the first case, if $x_1$ and, by induction, $X$ cancel out, then
$E={\anneau}Y{\anneau}^{-1}{\anneau}Y^{-1}{\anneau}^{-1}=1$.
In the second case, also by induction, $Y$ cancels out and
$E=x_1X{\anneau}{\anneau}^{-1}X^{-1}x_1^{-1}{\anneau}{\anneau}^{-1}=1$.
Therefore any $k$ or more participants can open the system.
Conversely, suppose only at most $k-1$ padlocks are opened.
Then by induction, $X$ can vanish, but not $Y$.
Further, $X$ can vanish only with the opening of at least $k-1$
padlocks. Therefore $x_1$ and $X$ cannot vanish simultaneously.
Thus either nothing vanishes or $E$ has one of two forms,
$E=X{\anneau}Y{\anneau}^{-1}X^{-1}{\anneau}Y^{-1}{\anneau}^{-1}$, with
non-vanishing $X$ and $Y$,
or
$E=x_1{\anneau}Y{\anneau}^{-1}x_1^{-1}{\anneau}Y^{-1}{\anneau}^{-1}$,
with non-vanishing $x_1$ and $Y$. In both cases, the system is not
opened.
We have proven that the system created by \cref{alg:knotted} is a
$k$-out-of-$n$ threshold system.

Now for the complexity bound with $n\geq{2}$.
Let $W(k,n)$ be the number of wrappings for a $k$-out-of-$n$ system
created by \cref{alg:knotted}.
We have that $W(2,2)=2$, the AND gate, and $W(1,2)=4$, the OR gate.
Next, for $k=n$, by the construction of Line~\ref{lin:noon},
we have that $W(n,n)=n$.
For $k=n-1$, by the construction of Line~\ref{lin:nmuoon},
we have that $W(n-1,n)=2n$.
For $k=1$, by the construction of Line~\ref{lin:1oon},
we have that $W(1,n)=2(1+W(1,n-1))$. This is
$W(1,n)=(\sum_{i=1}^{n-2} 2^i)+(2^{n-2}W(1,2))=\frac{3}{2}2^n-2$.

Otherwise, we have that $n\geq{4}$ and $k\in[2..(n-2)]$.
There, we first show by induction that $W(k,n)$
satisfies:
\begin{equation}\label{eq:twotoknu}
W(k,n)\geq{2^{n-k}}.
\end{equation}
Indeed, this is true for both
$W(2,2)=2\geq{2^{2-2}}$ and $W(1,2)=4\geq{2^{2-1}}$.
Next, for $k=n$, $W(n,n)=n\geq{1}=2^{n-n}$;
for $k=n-1$, $W(n-1,n)=2n\geq{2}=2^{n-(n-1)}$ and
for $k=1$, $W(1,n)=\frac{3}{2}2^n-2\geq{2^{n-1}}$.
Otherwise, $W(k,n)$ satisfies:
\begin{equation}\label{eq:wkn}
W(k,n)=2(3+W(k-1,n-1)+W(k,n-1)).
\end{equation}
Thus by induction and \cref{eq:wkn}, we can lower out the two
recursive calls, and have that
$W(k,n)\geq{2(3+2^{n-1-(k-1)}+2^{n-1-k})}\geq{2^{n-k}}$ and
\cref{eq:twotoknu} is proven.

Then, second, we refine this analysis by lowering $W(k,n-1)$ in
\cref{eq:wkn} with
\cref{eq:twotoknu}.
This gives $W(k,n)\geq{2(3+W(k-1,n-1))+2\cdot{2^{n-1-k}}}$.
Then we recurse for $W(k-1,n-1)$, to obtain that:
\[\begin{split}
W(k,n)&\geq{2(3+2(3+W(k-2,n-2)+2^{n-2-k}))+2^{n-k}}\\
&=\left(6\sum_{i=0}^{k-2}2^i\right)+2^{k-1}W(1,n-k+1)+2^{n-k}(k-1)\\
&=6(2^{k-1}-1)+2^{k-1}(\frac{3}{2}2^{n-k+1}-2)+2^{n-k}(k-1)\\
&=\frac{3}{2}\cdot{2^{n}}+2^{k+1}-6+2^{n-k}(k-1).
\end{split}\]
This concludes the proof, using the fact that
$\left(2^{k+1}+2^{n-k}(k-1)-6\right)\geq{6}$ for $k\geq{2}$ and $n-2\geq{k}$.
%
\end{proof}

To realize this solution in practice, one could for instance to use a
high security cable seal: once fastened the wire cannot be taken out
of the seizing device, see \cref{fig:cablelock}, left.
Before closing the seal, wrap it around the door latch clockwise;
then install the $k$-out-of-$n$ knotted padlock threshold system on
the wire;
finally wrap the wire around the door latch anticlockwise and then seal
it. This is shown in \cref{fig:cablelock}, right.

\begin{figure}[htbp]
\centering
\hfill
\includegraphics[height=4cm]{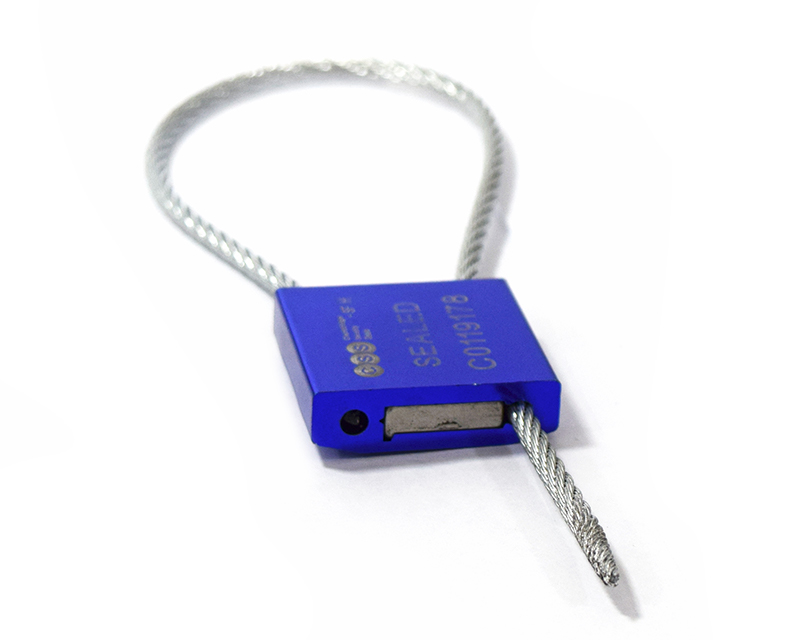}%
\hfill
\includegraphics[height=5cm]{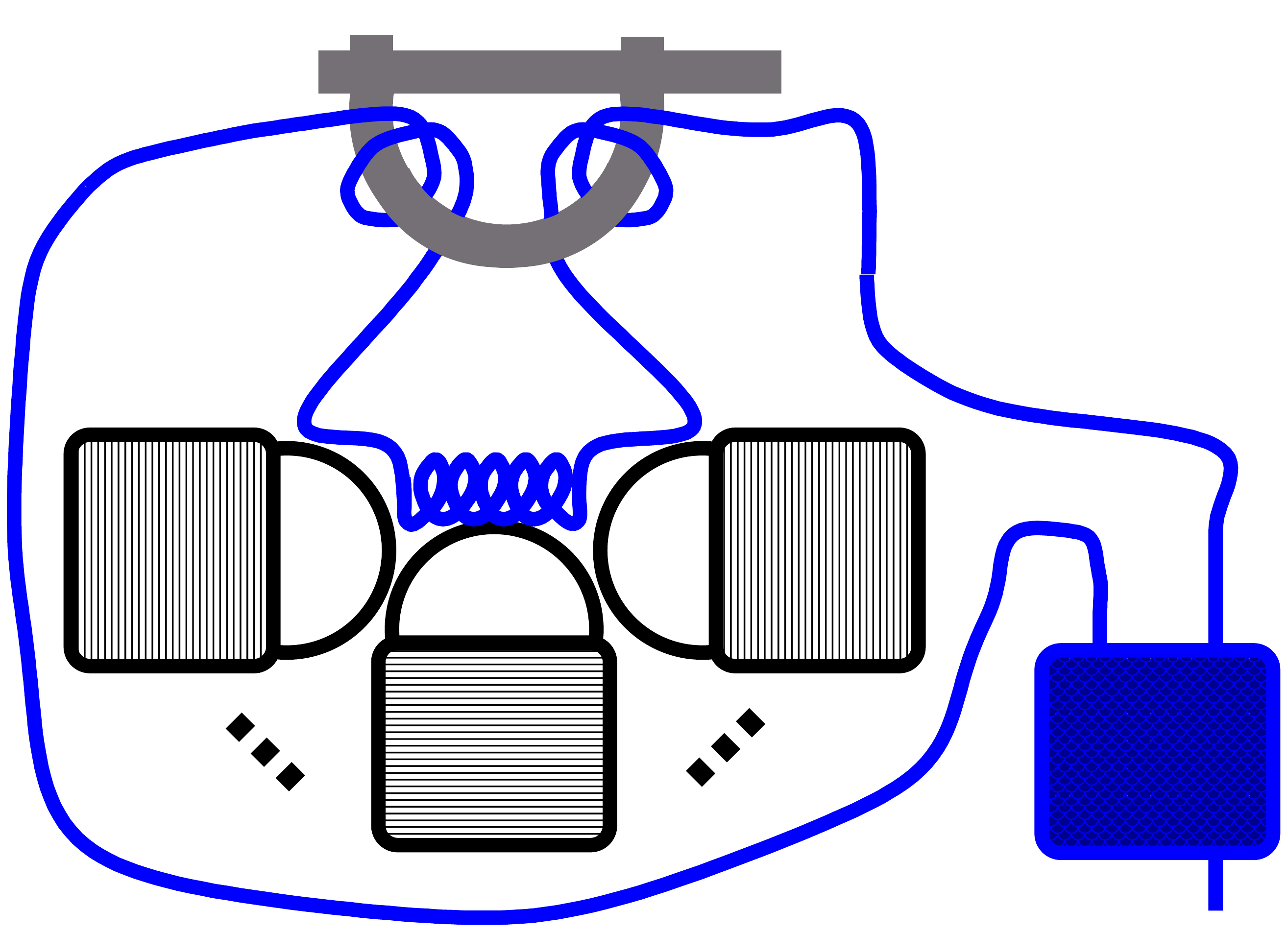}%
\hfill
\caption{Secure cable wire seal (left) and knotted padlock threshold system (right).}\label{fig:cablelock}
\end{figure}

Now, if the latch is smaller than the seizing part and than the
padlocks, the door cannot be opened unless all the locks are removed.
This can happen if at least $k$-out-of-$n$ participants open the
padlocks: then the other ones are freed by construction.

\begin{remark} Unless $k=n$ from the start, \cref{alg:knotted} will
  never encounter the case that the threshold equals the number of
  remaining participants. All the other cases perform exactly as
  many clockwise and anticlockwise wrappings. Therefore it is not
  mandatory to use a sealable cable. Any wire loop with a large enough
  part will do, for instance an already sealed cable. The setup is
  more cumbersome, but it is sufficient to pass a curl inside the
  latch or the ring to simultaneously wrap clockwise and
  anticlockwise.
\end{remark}

Thus, we have another possibility for a $k$-out-of-$n$ physical
threshold system with exactly $n$ padlocks. Unfortunately, we can make
it work only with an exponential number of wrappings. For instance,
\cref{alg:knotted} requires $279\,038$ initial wrappings for the
$6$-out-of-$11$ case. For now, the setup of this solution is therefore
not really practical.

%% file: sqrt.tex
\section{Square Root Bounds for Threshold Systems}\label{sec:sqrt}
We now have tools to deal with larger thresholds. First we give a
necessary condition for systems using less than~$n$ padlocks. Knotted
designs are not needed, but some small access structure arrangements
can help.
For instance, we can show that our physical device is optimal when~$k$
is larger than~$\sqrt{2n}$. For instance, we fully answer Liu's
question about the smallest number of locks needed to implement a
$6$-out-of-$11$ threshold system: this is~$11$ padlocks.
Then, the necessary condition, together with block design theory and
our padlock algebra of \cref{sec:algebra}, enables us to build
padlock systems with strictly less than~$n$ padlocks: for instance systems with
only about~$2.5\sqrt{n}$ padlocks for~$3$-out-of-$n$ thresholds.

\subsection{A Necessary Condition and the Answer to Liu's Problem}
We first begin with a necessary condition, analyzing the set
difference cardinality of their sets of keys.

\begin{restatable}{proposition}{propnec}\label{prop:nec}
$\forall{n}$ and $\forall{k}\geq{3}$, if a $k$-out-of-$n$ threshold
system uses strictly less than~$n$ padlocks, then
apart from participants owning the single key of a given
padlock, the other participants must satisfy:
\begin{enumerate}
\item The cardinality of their~$2$ by~$2$ set difference is bigger or equal to~$k-1$;
\item Each of them owns at least~$k$ distinct keys.
\end{enumerate}
\end{restatable}

\begin{proof}
Apart from participants owning the single key of a given
padlock, the others own only keys that are duplicated and owned by
several users. In the following we say that these duplicated keys
owned by several users are \emph{shared}, and we identify any
duplicated keys of the same padlock (we also thus say that shared keys
are \emph{reused} when we encounter the duplicate of an already
used key).

We thus consider the subgroup of participants owning only shared keys.
First, let $A, B, C$ be the sets of keys of three users from this subgroup, such
that $(A\setminus{B})\subset{C}$; then $(A\cup{B}\cup{C})=(B\cup{C})$.
More generally, suppose $|A\setminus{B}|=d$. Then these $d$ keys are
reused (as all the keys in $A$ are shared keys). Thus there exist $d' \leq d$ other participants with sets of keys~$C_j$ such that
$(A\setminus{B})\subset{\bigcup_{j=1}^{d'} C_j}$. Let
${\mathcal{C}}={\bigcup_{j=1}^{d'} C_j}$, then
$A\cup{B}\cup{\mathcal{C}}=B\cup{\mathcal{C}}$.
In other words, there is a group of~$d'+2$ users with
the same keys as a group  of~$d'+1$ users.
Therefore~$d'\geq{k-1}$ (and~$d \geq k-1$, as~$d \geq d'$): otherwise complete these~$d'+2$ participants
with~$k-d'-2$ others. Those~$k$ participants can open the door, as well
as the~$k-1$ participants obtained when removing~$A$. This would
contradict the fact that we have a~$k$-threshold system.
So we can restrict the analysis to groups of people having sets of
shared keys, with minimal~$2$ by~$2$ set difference cardinality larger
than~$k-1$.

Second, within such a group of participants owning only shared keys, suppose that one participant~$P$ owns a
number~$i$ of distinct keys strictly lower than the threshold
$k$. Then at least one of his keys cannot be reused. Otherwise there
exists a group of~$i' \leq i$ participants owning the same keys as
these~$i'$ participants plus the initial one~$P$. As~$i'<k$, complete
these~$i'+1$ participants with~$k-i'-1$ others. Those~$k$ participants
can open the door, as well as the~$k-1$ participants obtained when
removing~$P$ from the group. This would contradict the fact that we
have a~$k$-threshold system. %
We have proven: to build a threshold system with strictly less than
$n$ padlocks, apart from participants owning the single key of a given
padlock, the other participants must satisfy that
both,
their~$2$ by~$2$ set difference has more than~$k-1$ keys
and
each of them owns at least~$k$ distinct keys.
\end{proof}

We can now fully answer Liu's question with \cref{thm:2k},
$\ell_{6,11}=11$. The theorem also shows that our system is optimal
for all~$k\geq{\frac{\sqrt{8n+1}-1}{2}}$.
\begin{theorem}\label{thm:2k}
$\ell_{k,n}\geq\min\left\{n;\frac{k(k+1)}{2}\right\}$.
\end{theorem}
\begin{proof}
For~$k=2$, \cref{lem:nthree} gives the result.
Now, for~$k\geq{3}$, let~$G$ be the set of players only owning shared keys, let~$i$ be the number of these players and let~$t$ be the
number of padlocks in a~$k$-out-of-$n$ threshold system.
If~$i=0$, then~$n$ players own the single key of a padlock and
$t\geq{n}$.
If~$i=1$, then that player has at least~$k$ new distinct and shared keys by
\cref{prop:nec}. Those shared keys are by definition not among the
singly owned keys and thus~$t=n-1+k\geq{n}$.
More generally, if~$i\leq{k}$, then~$t\geq{n-i+k}\geq{n}$.

Now for~$i>k$, one of the~$i$ players has at least~$k$ distinct, but
shared, keys.
Then the next participant has also all his shared keys not among the
singly used, and at least~$k-1$ keys not shared with the previous
player (otherwise their set difference is not larger than~$k-1$).
More generally, let~$B_1$,~$\ldots$,~$B_h$ and~$A$ be the sets of keys of
$h+1$ distinct members of~$G$, with~$h\leq{k}$.
Let~${\mathcal{B}}=\bigcup_{j=1}^h B_j$ and let
$|A\setminus{\mathcal{B}}|=d$. Those~$d$ keys are shared (as all the
keys in~$A$) and, therefore, there exist~$d' \leq d$ other participants with
sets of keys~$C_j$
such that $(A\setminus{\mathcal{B}})\subset{\bigcup_{j=1}^{d'} C_j}$.
Let~${\mathcal{C}}={\bigcup_{j=1}^{d'} C_j}$, then
$A\cup{\mathcal{B}}\cup{\mathcal{C}}={\mathcal{B}}\cup{\mathcal{C}}$.
But then~$1+h+d'>k$: otherwise a group of~$k$ participants has the same set
of keys as a group of~$k-1$. In other words, we have shown that
$d'>k-1-h$, or more precisely that~$d'\geq{k-h}$ (and thus~$d\geq{k-h}$ as~$d\geq{d'}$).
Therefore, up to the~$k$-th person in group~$G$ (a group of~$i$
participants with~$i>k$), each person must have
at least~$k-h$ keys not in the sets of the~$h$ previous ones.
Since~$i>k$, this is at least $k+k-1+\sum_{h=2}^k (k-h)=k(k+1)/2$ keys.
Then we have that the total number of keys satisfies
$t\geq{n-i+k(k+1)/2}\geq{k(k+1)/2}$.
\end{proof}

\subsection[Packings, Johnson Bound and a 3-threshold Realization for
up to 12 Participants with only 9 Padlocks]{Packings, Johnson Bound
  and a~$3$-threshold Realization for up to~$12$ Participants with
  only~$9$ Padlocks}\label{ssec:nine}
A sufficient condition to satisfy
\cref{prop:nec} for a~$3$-threshold system
with less than~$n$ padlocks is that a given pair of keys is never
given to more than one person.
Indeed, then, two persons never share a pair of keys and thus if they each
own more than two keys, then their set difference is at least~$2=k-1$.

This is thus sufficient for such a system to contain a~$(2,1)-$packing, as defined
thereafter:
\begin{definition}[See e.g.,~\cite{Chee:2013:packings}]
Let~$t$,~$k$, and~$p$ be integers with~$t > k > p\geq{2}$. Let
$\lambda$ be a positive integer.
A {\em~$(p,\lambda)-$packing} of order~$t$, and blocksize~$k$ is a set V of
$t$ elements, and a collection B of~$k$-element subsets (blocks) of V,
so that every~$p$-subset of V appears in at most~$\lambda$ blocks.
\end{definition}

With this, we have Johnson's bound~\cite{Johnson:1962:bound}, that
states that a maximal packing has a number of blocks upper bounded by:
\begin{equation}\label{eq:johnson}
\left\lfloor\frac{t}{k}\left\lfloor\frac{t-1}{k-1}\right\rfloor\right\rfloor.
\end{equation}

\Cref{eq:johnson} then suggests that systems with
$t={\mathcal{O}}(k\sqrt{n})$ padlocks might be possible.

Unfortunately, \cref{prop:nec} is probably not sufficient itself:
it might be possible to fulfill its conditions while still having some set
of players of size strictly lower than~$k$ having the same set of keys
as some set of players of size~$k$. 
However, we can at least prove that for~$k=3$ we can
always use Steiner triad systems to build~$3$-threshold systems.
A Steiner triad system is a $(2,1)$-packing with blocksize~$3$.
In other words, it is a pair of sets such that every pair of
elements of the first set appears together in a unique {\em triad}
(or a triangle, or a triplet) of the second one~\cite{Bose:1939:bibd}.
As a consequence, it is possible to build a~$3$-threshold padlock system
with only~${\mathcal O}(\sqrt{n})$ padlocks:
\begin{enumerate}
\item in a Steiner triad system for a set of keys, no pair of keys is shared by two
  triads; therefore giving a triad of keys to each player will
  satisfy the necessary~\cref{prop:nec};
\item then, the following~\cref{prop:steiner} shows that for the
  particular case of $k=3$ this is also sufficient;
\item finally, with Johnson's bound, a Steiner triad system with
$t={\mathcal O}(\sqrt{n})$ will have sufficiently many triads
to give one to each of the $n$ players.
\end{enumerate}

\begin{restatable}{proposition}{propsteiner}\label{prop:steiner}
Any Steiner triad system gives rise to a~$3$-threshold system.
\end{restatable}

\begin{proof}
By construction, a Steiner triad system satisfies the
necessary condition of \cref{prop:nec}.  Second, in order
to use it as a~$3$-threshold system, we need to differentiate triples
of triads from pairs of triads (as a two participants should not be able to open the door, but three participants should).
On the one hand, all triples that have~$7$, or more,
distinct values all together, cannot be equated by pairs of triads.
On the other hand, by the condition
on pairs of elements being uniquely found in a single triad, triples of triads
have at least~$6$ distinct values overall.  So the only remaining case
is to prove that the~$6$ distinct values of triples of triads with
only~$6$ distinct values, in any construction, cannot be found in
pairs of triads of the system.

To have only~$6$ distinct values, any two of the triple of triads
must share one value, and the third one must share a value with each
of the two others.
W.l.o.g., this is triads~$(a,b,c);(a,d,e);(b,d,f)$, with distinct
values~$a,b,c,d,e,f$. Now suppose that these~$6$ values are contained
in a pair of triads. Then, among~$a,b,c$, at least two of them must
be in one of the pair. But by the unicity of triads containing a
given pair this means that~$(a,b,c)$ {\em is} one of the pairs. The other
pair must now be~$(d,e,f)$. But the triad~$(a,d,e)$ is in the
system so the pair~$(d,e)$ is shared by two different blocks. This is
a contradiction and no pair of triads can share the~$6$ distinct
values of a triple.
\end{proof}
Finally, by setting up a minimal Steiner system for any number of
players, for instance using a Bose construction~\cite{Bose:1939:bibd},
we have the following \cref{alg:threeoutofn} to setup a~$3$-out-of-$n$
system. This provides an upper bound of~${\mathcal O}(\sqrt{n})$
for the number of padlocks for such a system. How to open such a
system is then described in~\cref{alg:freethreeoutofn}.

\begin{algorithm}[htbp]
\caption{Bose Three-out-of-$n$ threshold system with shared
  keys}\label{alg:threeoutofn}
\begin{algorithmic}[1]
\Require $n\geq{2}$.
\Ensure A $3$-out-of-$n$ threshold system with
$t=6\left\lceil\frac{\sqrt{24n+1}-5}{12}\right\rceil+3$
padlocks.
\State Let
$v=\displaystyle\left\lceil\frac{\sqrt{24n+1}-5}{12}\right\rceil$, $m=2v+1$ and $t=6v+3$;
\State Setup a $7$-out-of-$t$ threshold system;
\LeftComment{Bose construction~\cite{Bose:1939:bibd}}
\For{$x=1..m$}\label{lin:startbose}
\State Give the triad of keys $<x,x+m,x+2m$ to next player;
\EndFor
\For{$x=1..m$}
\For{$y=1..m$}
\For{$k=0..2$}
\State $a=x+km$;
\State $b=y+km$;
\State $c=((x+y)2^{-1}-1)\mod{m}$;
\State $d=((k+1)\mod{3})$;
\State Give the triad of keys $<a,b,1+c+dm>$ to next player;
\EndFor
\EndFor
\EndFor\label{lin:endbose}
\LeftComment{CNF for the particular triads with only $6$ distinct values}
\State $h=0$;
\For{each triple of triads of keys}
\If{this triple of triads contains only $6$ distinct keys}
\State Setup a conjonctive clause with these $6$ values;
\State increment $h$;
\EndIf
\EndFor
\LeftComment{Setup a DNF via~\cref{alg:disjunctive}}
\State Setup a $1$-out-of-$(h+1)$ device \Comment{one latch for each
conjonction, plus an additional one}
\For{each of the $t$ padlocks}
\State Pass a chain through the hole of each latch corresponding to a
conjunction containing that padlock;
\State Pass the chain through the hole of one free latch of the
$7$-out-of-$t$ device;
\State Close the chain with that padlock.
\EndFor
\State Attach the $7$-out-of-$t$ device to remaining latch of the
$1$-out-of-$(h+1)$ device.
\end{algorithmic}
\end{algorithm}

\begin{algorithm}[htb]
\caption{Opening the device of \cref{alg:threeoutofn}}\label{alg:freethreeoutofn}
\begin{algorithmic}[1]
\Require $n\geq{2}$;
\Require A $3$-out-of-$n$ device made by~\cref{alg:threeoutofn};
\Require The $3$ sets of keys give to $3$ participants
by~\cref{alg:threeoutofn}.
\Ensure The system is opened.
\If{the $3$ users own more than $7$ distinct keys}
\State They open the $7$-out-of-$t$ device, which then opens the
$1$-out-of-$(h+1)$ DNF.
\Else\Comment{Thus they own one of the $h$ groups of $6$ distinct keys}
\State They open their padlocks, this frees one clause of the DNF.
\EndIf
\end{algorithmic}
\end{algorithm}

\begin{restatable}{theorem}{thmbose}\label{thm:bose}
\cref{alg:threeoutofn} is correct and
\[\ell_{3,n} \leq
\min\left\{
  6\displaystyle\left\lceil\frac{\sqrt{24n+1}-1}{12}\right\rceil+1;
  6\displaystyle\left\lceil\frac{\sqrt{24n+1}-5}{12}\right\rceil+3\right\}.
\]
\end{restatable}

\begin{proof}
Any construction of a Steiner triad block
design works.
For instance, the Bose construction~\cite{Bose:1939:bibd} provides
such a design for any~$t=6\nu+3$. It is given in lines
\ref{lin:startbose} to \ref{lin:endbose} of \cref{alg:threeoutofn}.
\cref{prop:steiner} proves that these constructions can be
used as~$3$-threshold systems: use a~$7$-out-of-$t$ design and a large
DNF with all the possible groups of~$6$ distinct values never attained
by pairs of participants.
Further, the Bose construction attain the bound of \Cref{eq:johnson} for
$t=6\nu+3$ and~$k=3$, that
is~$t/3(t-1)/2=(2\nu+1)(3\nu+1)$. Thus for~$n$ players
with~$n\leq(2\nu+1)(3\nu+1)$ one can set up a Bose construction with
$t=6\nu+3$ and discard the blocks between~$n+1$ and~$t(t-1)/6$.
In other words, for a given~$n$, use
$\nu=\left\lceil\frac{\sqrt{24n+1}-5}{12}\right\rceil$ and only
$t=6\nu+3$ padlocks. This proves that \cref{alg:threeoutofn} is
correct.

To achieve the sometimes slightly better bound of the theorem, one needs
to use the Steiner triad system construction by
Skolem~\cite[Lemma 2.5]{Colbourn:1999:triple}.
There use~$t'=6\mu+1$ and~$t'/3(t'-1)/2=(6\mu+1)\mu\geq{n}$, so that now
$\mu=\left\lceil\frac{\sqrt{24n+1}-1}{12}\right\rceil$ and use
$t'=6\mu+1$ padlocks. Without the ceilings,~$t=t'$, for any~$n$. But
like this,~$t'$ and~$t$ are alternatively slightly better than the
other (a difference of~$2$ or~$4$).
\end{proof}

%% file: lk11_zones.tex
Finally, \cref{fig:lk11_zones} summarizes our current knowledge on the example
of $n=11$.

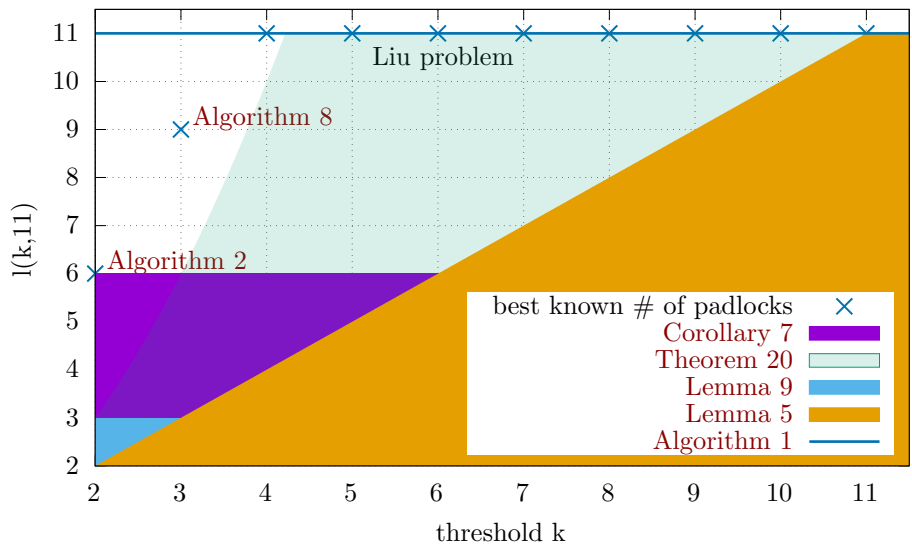
\begin{figure}[htbp]
\centering%
\input{padlock-gnuplottex-fig1}
\caption{n=11: number of padlocks for known k-out-of-11 algorithms
  (only the points above the regions are attainable).}\label{fig:lk11_zones}%
\end{figure}%

%% file: padlock-gnuplottex-fig1.tex
\begingroup
  \makeatletter
  \providecommand\color[2][]{%
    \GenericError{(gnuplot) \space\space\space\@spaces}{%
      Package color not loaded in conjunction with
      terminal option `colourtext'%
    }{See the gnuplot documentation for explanation.%
    }{Either use 'blacktext' in gnuplot or load the package
      color.sty in LaTeX.}%
    \renewcommand\color[2][]{}%
  }%
  \providecommand\includegraphics[2][]{%
    \GenericError{(gnuplot) \space\space\space\@spaces}{%
      Package graphicx or graphics not loaded%
    }{See the gnuplot documentation for explanation.%
    }{The gnuplot epslatex terminal needs graphicx.sty or graphics.sty.}%
    \renewcommand\includegraphics[2][]{}%
  }%
  \providecommand\rotatebox[2]{#2}%
  \@ifundefined{ifGPcolor}{%
    \newif\ifGPcolor
    \GPcolortrue
  }{}%
  \@ifundefined{ifGPblacktext}{%
    \newif\ifGPblacktext
    \GPblacktexttrue
  }{}%
  \let\gplgaddtomacro\g@addto@macro
  \gdef\gplbacktext{}%
  \gdef\gplfronttext{}%
  \makeatother
  \ifGPblacktext
    \def\colorrgb#1{}%
    \def\colorgray#1{}%
  \else
    \ifGPcolor
      \def\colorrgb#1{\color[rgb]{#1}}%
      \def\colorgray#1{\color[gray]{#1}}%
      \expandafter\def\csname LTw\endcsname{\color{white}}%
      \expandafter\def\csname LTb\endcsname{\color{black}}%
      \expandafter\def\csname LTa\endcsname{\color{black}}%
      \expandafter\def\csname LT0\endcsname{\color[rgb]{1,0,0}}%
      \expandafter\def\csname LT1\endcsname{\color[rgb]{0,1,0}}%
      \expandafter\def\csname LT2\endcsname{\color[rgb]{0,0,1}}%
      \expandafter\def\csname LT3\endcsname{\color[rgb]{1,0,1}}%
      \expandafter\def\csname LT4\endcsname{\color[rgb]{0,1,1}}%
      \expandafter\def\csname LT5\endcsname{\color[rgb]{1,1,0}}%
      \expandafter\def\csname LT6\endcsname{\color[rgb]{0,0,0}}%
      \expandafter\def\csname LT7\endcsname{\color[rgb]{1,0.3,0}}%
      \expandafter\def\csname LT8\endcsname{\color[rgb]{0.5,0.5,0.5}}%
    \else
      \def\colorrgb#1{\color{black}}%
      \def\colorgray#1{\color[gray]{#1}}%
      \expandafter\def\csname LTw\endcsname{\color{white}}%
      \expandafter\def\csname LTb\endcsname{\color{black}}%
      \expandafter\def\csname LTa\endcsname{\color{black}}%
      \expandafter\def\csname LT0\endcsname{\color{black}}%
      \expandafter\def\csname LT1\endcsname{\color{black}}%
      \expandafter\def\csname LT2\endcsname{\color{black}}%
      \expandafter\def\csname LT3\endcsname{\color{black}}%
      \expandafter\def\csname LT4\endcsname{\color{black}}%
      \expandafter\def\csname LT5\endcsname{\color{black}}%
      \expandafter\def\csname LT6\endcsname{\color{black}}%
      \expandafter\def\csname LT7\endcsname{\color{black}}%
      \expandafter\def\csname LT8\endcsname{\color{black}}%
    \fi
  \fi
    \setlength{\unitlength}{0.0500bp}%
    \ifx\gptboxheight\undefined%
      \newlength{\gptboxheight}%
      \newlength{\gptboxwidth}%
      \newsavebox{\gptboxtext}%
    \fi%
    \setlength{\fboxrule}{0.5pt}%
    \setlength{\fboxsep}{1pt}%
\begin{picture}(7200.00,4320.00)%
    \gplgaddtomacro\gplbacktext{%
      \csname LTb\endcsname
      \put(596,652){\makebox(0,0)[r]{\strut{}$2$}}%
      \csname LTb\endcsname
      \put(596,1014){\makebox(0,0)[r]{\strut{}$3$}}%
      \csname LTb\endcsname
      \put(596,1377){\makebox(0,0)[r]{\strut{}$4$}}%
      \csname LTb\endcsname
      \put(596,1739){\makebox(0,0)[r]{\strut{}$5$}}%
      \csname LTb\endcsname
      \put(596,2102){\makebox(0,0)[r]{\strut{}$6$}}%
      \csname LTb\endcsname
      \put(596,2464){\makebox(0,0)[r]{\strut{}$7$}}%
      \csname LTb\endcsname
      \put(596,2827){\makebox(0,0)[r]{\strut{}$8$}}%
      \csname LTb\endcsname
      \put(596,3189){\makebox(0,0)[r]{\strut{}$9$}}%
      \csname LTb\endcsname
      \put(596,3551){\makebox(0,0)[r]{\strut{}$10$}}%
      \csname LTb\endcsname
      \put(596,3914){\makebox(0,0)[r]{\strut{}$11$}}%
      \csname LTb\endcsname
      \put(708,448){\makebox(0,0){\strut{}$2$}}%
      \csname LTb\endcsname
      \put(1354,448){\makebox(0,0){\strut{}$3$}}%
      \csname LTb\endcsname
      \put(2000,448){\makebox(0,0){\strut{}$4$}}%
      \csname LTb\endcsname
      \put(2645,448){\makebox(0,0){\strut{}$5$}}%
      \csname LTb\endcsname
      \put(3291,448){\makebox(0,0){\strut{}$6$}}%
      \csname LTb\endcsname
      \put(3937,448){\makebox(0,0){\strut{}$7$}}%
      \csname LTb\endcsname
      \put(4583,448){\makebox(0,0){\strut{}$8$}}%
      \csname LTb\endcsname
      \put(5229,448){\makebox(0,0){\strut{}$9$}}%
      \csname LTb\endcsname
      \put(5874,448){\makebox(0,0){\strut{}$10$}}%
      \csname LTb\endcsname
      \put(6520,448){\makebox(0,0){\strut{}$11$}}%
      \csname LTb\endcsname
      \put(2040,3954){\makebox(0,0)[l]{\strut{}}}%
      \csname LTb\endcsname
      \put(2685,3954){\makebox(0,0)[l]{\strut{}}}%
      \csname LTb\endcsname
      \put(3331,3954){\makebox(0,0)[l]{\strut{}}}%
      \csname LTb\endcsname
      \put(3977,3954){\makebox(0,0)[l]{\strut{}}}%
      \csname LTb\endcsname
      \put(4623,3954){\makebox(0,0)[l]{\strut{}}}%
      \csname LTb\endcsname
      \put(5269,3954){\makebox(0,0)[l]{\strut{}}}%
      \csname LTb\endcsname
      \put(5914,3954){\makebox(0,0)[l]{\strut{}}}%
      \csname LTb\endcsname
      \put(6560,3954){\makebox(0,0)[l]{\strut{}}}%
      \csname LTb\endcsname
      \put(2807,3733){\makebox(0,0)[l]{\strut{}Liu problem}}%
      \csname LTb\endcsname
      \put(748,2142){\makebox(0,0)[l]{\strut{}}}%
      \csname LTb\endcsname
      \put(805,2192){\makebox(0,0)[l]{\strut{}\cref{alg:twooutofn}}}%
      \csname LTb\endcsname
      \put(1394,3229){\makebox(0,0)[l]{\strut{}}}%
      \csname LTb\endcsname
      \put(1451,3280){\makebox(0,0)[l]{\strut{}\cref{alg:threeoutofn}}}%
    }%
    \gplgaddtomacro\gplfronttext{%
      \csname LTb\endcsname
      \put(186,2373){\rotatebox{-270}{\makebox(0,0){\strut{}l(k,11)}}}%
      \csname LTb\endcsname
      \put(3775,142){\makebox(0,0){\strut{}threshold k}}%
      \csname LTb\endcsname
      \put(5978,1855){\makebox(0,0)[r]{\strut{}best known \# of padlocks}}%
      \csname LTb\endcsname
      \put(5978,1651){\makebox(0,0)[r]{\strut{}\cref{cor:sperner}}}%
      \csname LTb\endcsname
      \put(5978,1447){\makebox(0,0)[r]{\strut{}\cref{thm:2k}}}%
      \csname LTb\endcsname
      \put(5978,1243){\makebox(0,0)[r]{\strut{}\cref{lem:nthree}}}%
      \csname LTb\endcsname
      \put(5978,1039){\makebox(0,0)[r]{\strut{}\cref{lem:ktwo}}}%
      \csname LTb\endcsname
      \put(5978,835){\makebox(0,0)[r]{\strut{}\cref{alg:koutofnphys}}}%
    }%
    \gplbacktext
    \put(0,0){\includegraphics[width={360.00bp},height={216.00bp}]{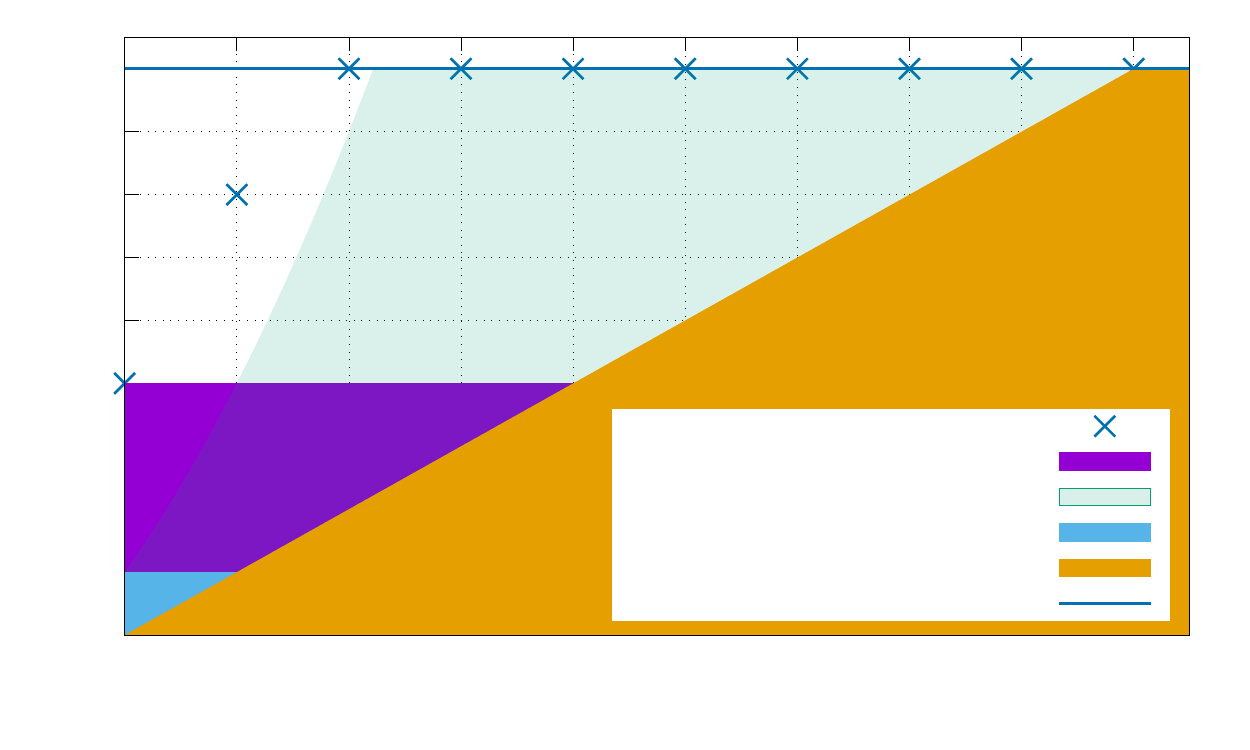}}%
    \gplfronttext
  \end{picture}%
\endgroup

%% file: cex.tex
\subsection[3-threshold realizations with fewer than n
padlocks]{$3$-threshold realizations with fewer than $n$
  padlocks}\label{app:cex}
We give the smallest example realizing \cref{prop:steiner}:
a $3$-out-of-$12$ system (thus also a $3$-out-of-$11$ system), with
only $9$ padlocks, $36$ keys and $82$ latches, and an example using
normal forms to reduce the number of latches for a $3$-out-of-$13$
system with only $11$ padlocks, $36$ keys and $33$ latches. Indeed,
consider the first terms of \Cref{eq:johnson} for $k=3$ and
$t=0,1,2,\ldots$, this is:
$0, 0, 0, 1, 1, 3, 4, 7, 8, 12, 13, 18, 20, 26, 28,\ldots$~\cite{oeis:johnson}.

\subsubsection[An example realization]{An example realization of \cref{alg:threeoutofn}}

The smallest $t$ such that \Cref{eq:johnson} is strictly
larger than $t$ is for $t=9$ with a bound of $12$ subsets.
Hence, packing with $9$ padlocks is realizable, for instance
with the Steiner triads of keys given in \cref{tab:nine}.

\begin{table}[htbp]
\centering\small
\caption{A maximal $(2,1)-$ packing of order $9$ and blocksize $3$. It has
  $12$ blocks.}\label{tab:nine}
\begin{tabular}{crrrrrrrrrrrr}
\hline
1 & 2 & 3 & 1 & 4 & 3 & 1 & 4 & 2 & 2 & 5 & 1\\
4 & 5 & 6 & 2 & 5 & 7 & 3 & 6 & 7 & 3 & 6 & 8\\
7 & 8 & 9 & 6 & 9 & 8 & 5 & 8 & 9 & 4 & 7 & 9\\
\hline
\end{tabular}%
\end{table}

By inspection, there are $72$ triples of triads (so $3$ participants
owning each $3$ keys) with only $6$ distinct keys (for instance the
triads $<1,2,3>$, $<1,4,8>$, $<2,4,7>$).  The $148$ other triples of
triads have at least $7$ distinct keys (if a triple have a total of
less than $5$ distinct keys it would mean that at least two of them
share a pair).
But \cref{prop:steiner} shows that none of the $72$ sets of six keys obtained with
three triads can be obtained with only a pair of
triads (for instance the triple $(1,2,3);(1,4,8);(2,4,7)$ contains
only the six distinct keys $1,2,3,4,7,8$, but no pair of triples
contain the same set of keys).
The latter ensures that no subset of $2$ participants can unlock the
door. Further, all these $72$ sets of $6$ keys are distinct.

Therefore, it is possible to set up a $3$-out-of-$12$ system using
only $9$ padlocks. The idea of \cref{alg:threeoutofn} is that either a
group owns $7$ distinct keys or it owns one of the $72$ sets of $6$
keys not reachable by a pair of participants.
Overall, that solution uses $9$ padlocks, $9$ chains, $36$ keys,
a $7$-out-of-$9$ and a $1$-out-of-$73$ design (that is $9+73=82$ latches).
The following process gives the instance of \cref{alg:threeoutofn} for
this system:
\begin{enumerate}
\item Set up $9$ padlocks and make $4$ copies of each key;
\item Give $3$ keys to each of the $12$ participants following the packing
  of \cref{tab:nine};
\item Set up a $1$-out-of-$73$ design;
\item Set up a $7$-out-of-$9$ design and attach it to one the latches
  of the $1-73$ design;
\item Use \cref{alg:disjunctive} to complete the $72$ other
  latches: pass a chain through the hole of each latch corresponding
  to a disjunction containing that key; close that chain with the
  associated padlock.
\end{enumerate}

\subsubsection[Bose bound is not enough]{The bound of \cref{thm:bose} is not enough}
Next, we give a small example where there exists a shortcut to use less
latches than with the latter construction. We use some results of
\cref{sec:algebra} to help for the construction.
For $13$ participants, \cref{thm:bose} would provide a system with
either
$15=6\left\lceil\frac{\sqrt{24*13+1}-5}{12}\right\rceil+3$ or
$13=6\left\lceil\frac{\sqrt{24*13+1}-1}{12}\right\rceil+1$ 
padlocks. This is already not better than $13$ padlocks, directly
attainable with our $3$-out-of-$13$ device.

But we even show next a $3$-threshold system for $12$ or $13$ participants
with only $11$ padlocks, $36$ or $39$ keys and only $5$ additional devices for
a total of $33$ latches. We give in \cref{tab:realize},
afterwards, a realization of a packing with $3$-subsets.
Then we proceed by inspection of the triples and pairs of triads of
keys.
\begin{table}[htbp]\centering\small
\caption{Distribution of $11$ keys to $13$ participants without any
  reused pair.}\label{tab:realize}
\begin{tabular}{crrrrrrrrrrrrr}
\hline
Player & 1 & 2 & 3 & 4 & 5 & 6 & 7 & 8 & 9 & 10 & 11 & 12 & 13\\
\hline
\multirow{3}{*}{keys}
& 1 & 1 & 1 & 1 & 1 & 2 & 2 & 2 & 2 & 3 & 3 & 3 & 3\\
& 2 & 4 & 6 & 8 & 10& 4 & 5 & 8 & 9 & 4 & 5 & 8 & 9\\
& 3 & 5 & 7 & 9 & 11& 6 & 7 & 10& 11& 7 & 6 & 11& 10\\
\hline
\end{tabular}%
\end{table}
There are $\binom{13}{3}=286$ triples of triads and among them
$56$ have only $6$ distinct keys. All the other triples have at least
$7$ distinct keys.
Also, there are $\binom{13}{2}=78$ pairs of triads and among them
$24$ have exactly $6$ distinct keys. All the other pairs have at most
$5$ distinct keys.
Further, on the one hand, all those $24$ pairs contain no more and no
less than $2$ keys among $8,9,10,11$.
On the other hand, among the $56$ triples either they contain more
than $3$ keys among $8,9,10,11$ or their $6$ distinct keys are lower
than $7$. This is summarized by \Cref{eq:cex}.
\begin{multline}\label{eq:cex}
(7~\text{out-of}~1{\ldots}11)
~~~\text{\bf OR}~~~\\
\left(~~ (6~\text{out-of}~1{\ldots}11)
 \quad \text{\bf AND } \quad
 ((3~\text{out-of}~8{\ldots}11) ~\text{\bf OR}~ (5~\text{out-of}~1{\ldots}7) )
\right)
\end{multline}

So, by luck, the following construction realizes a $3$-threshold
system for $13$ participants with $11$ padlocks.
We need a $7$-out-of-$11$ device as well as a $6$-out-of-$11$, a
$5$-out-of-$7$ and a $3$-out-of-$4$ of our designs. Finally a
classical $2$-out-of-$2$ device is needed for the {\bf AND} part.
All of these are organized as follows, in order to realize the formula
of \Cref{eq:cex}.

\begin{figure}[htbp]\centering
\includegraphics[width=0.85\textwidth]{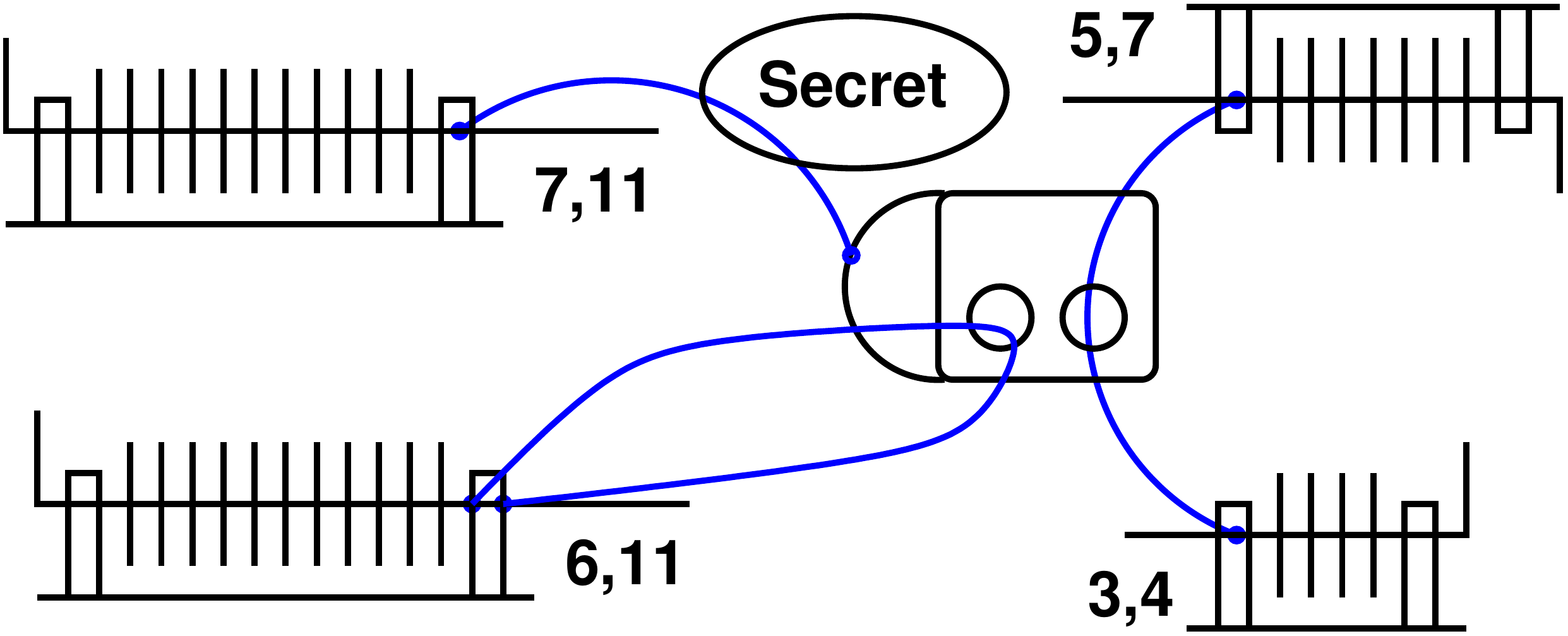}
\caption{A $3$-threshold realization
for $13$ participants with $11$ padlocks.
Each participant owns $3$ keys with the distribution of
\cref{tab:realize}.
On the one hand,
any~$3$ participants have either at
least~$7$ distinct keys or if they have only $6$ keys then they have
at least~$5$ for padlocks numbered~$1$ to~$7$ or at least~$3$ for
padlocks numbered~$8$ to~$11$.
On the other hand, no pair of participants has a total of $6$ distinct
keys and either $5$ of the first seven ones or $3$ for the last four
ones.}\label{fig:realize}%
\end{figure}

Each of the eleven padlocks is used once to close a chain as in
\cref{alg:disjunctive}. For each padlock its associated chain
will go through the hole of each of up to the four devices
(the devices $(7,11)$ and $(6,11)$ have each~$11$ latches so are linked to
all the padlock; while device $(5,7)$ is for the padlocks numbered~$1$
to~$7$ and device $(3,4)$ is for the padlocks numbered~$8$ to~$11$).
This will realize the disjunctions {\bf OR} in
\Cref{eq:cex}.
Finally, the disjunction of the devices $(5,7)$ and $(3,4)$ is
linked via a chain, and that together with the $(6,11)$ device are
associated via a $2$-out-of-$2$ device, as in
\cref{alg:conjunctive}.
The whole system is shown in \cref{fig:realize}.
Overall, it requires fewer padlocks, but quite a bunch of other
devices.

The same system works also for a $3$-threshold realization
for $12$ participants with $11$ padlocks.
Just use the $12$ first triads of keys of~\cref{tab:realize}
with the same system. Yet this solution uses more padlocks than \cref{alg:threeoutofn}.

\cref{fig:l3n_zones} summarizes what we know for $3$-out-of-$n$
systems. We see that for a threshold of three the minimal number of
padlocks is in between ${\mathcal O}(\log(n))$ and
${\mathcal O}(\sqrt{n})$.

\input{l3n_zones}

%% file: l3n_zones.tex
\begin{figure}[htbp]
\centering%
\input{padlock-gnuplottex-fig2}
\caption{k=3: number of padlocks for known 3-out-of-n algorithms (only
  the points below the algorithms lines and above the regions could use less padlocks).}\label{fig:l3n_zones}%
\end{figure}
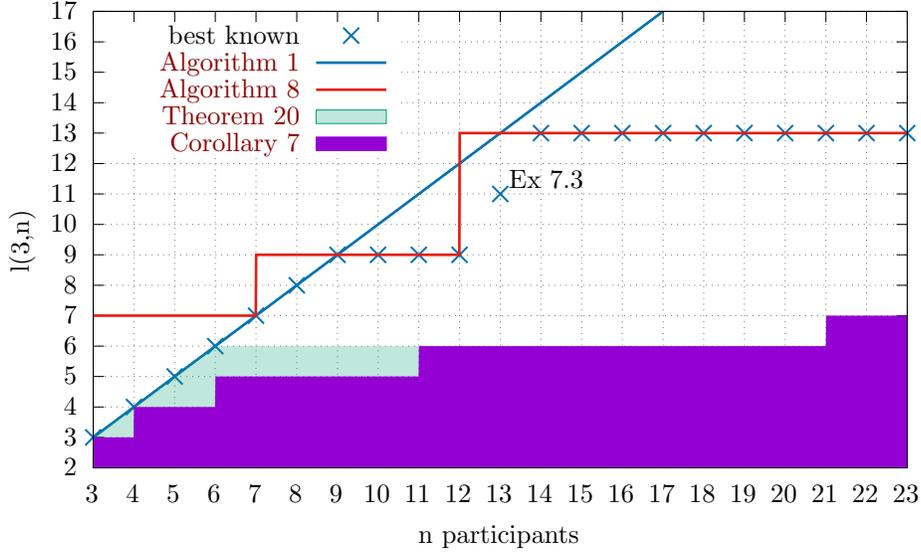%

%% file: padlock-gnuplottex-fig2.tex
\begingroup
  \makeatletter
  \providecommand\color[2][]{%
    \GenericError{(gnuplot) \space\space\space\@spaces}{%
      Package color not loaded in conjunction with
      terminal option `colourtext'%
    }{See the gnuplot documentation for explanation.%
    }{Either use 'blacktext' in gnuplot or load the package
      color.sty in LaTeX.}%
    \renewcommand\color[2][]{}%
  }%
  \providecommand\includegraphics[2][]{%
    \GenericError{(gnuplot) \space\space\space\@spaces}{%
      Package graphicx or graphics not loaded%
    }{See the gnuplot documentation for explanation.%
    }{The gnuplot epslatex terminal needs graphicx.sty or graphics.sty.}%
    \renewcommand\includegraphics[2][]{}%
  }%
  \providecommand\rotatebox[2]{#2}%
  \@ifundefined{ifGPcolor}{%
    \newif\ifGPcolor
    \GPcolortrue
  }{}%
  \@ifundefined{ifGPblacktext}{%
    \newif\ifGPblacktext
    \GPblacktexttrue
  }{}%
  \let\gplgaddtomacro\g@addto@macro
  \gdef\gplbacktext{}%
  \gdef\gplfronttext{}%
  \makeatother
  \ifGPblacktext
    \def\colorrgb#1{}%
    \def\colorgray#1{}%
  \else
    \ifGPcolor
      \def\colorrgb#1{\color[rgb]{#1}}%
      \def\colorgray#1{\color[gray]{#1}}%
      \expandafter\def\csname LTw\endcsname{\color{white}}%
      \expandafter\def\csname LTb\endcsname{\color{black}}%
      \expandafter\def\csname LTa\endcsname{\color{black}}%
      \expandafter\def\csname LT0\endcsname{\color[rgb]{1,0,0}}%
      \expandafter\def\csname LT1\endcsname{\color[rgb]{0,1,0}}%
      \expandafter\def\csname LT2\endcsname{\color[rgb]{0,0,1}}%
      \expandafter\def\csname LT3\endcsname{\color[rgb]{1,0,1}}%
      \expandafter\def\csname LT4\endcsname{\color[rgb]{0,1,1}}%
      \expandafter\def\csname LT5\endcsname{\color[rgb]{1,1,0}}%
      \expandafter\def\csname LT6\endcsname{\color[rgb]{0,0,0}}%
      \expandafter\def\csname LT7\endcsname{\color[rgb]{1,0.3,0}}%
      \expandafter\def\csname LT8\endcsname{\color[rgb]{0.5,0.5,0.5}}%
    \else
      \def\colorrgb#1{\color{black}}%
      \def\colorgray#1{\color[gray]{#1}}%
      \expandafter\def\csname LTw\endcsname{\color{white}}%
      \expandafter\def\csname LTb\endcsname{\color{black}}%
      \expandafter\def\csname LTa\endcsname{\color{black}}%
      \expandafter\def\csname LT0\endcsname{\color{black}}%
      \expandafter\def\csname LT1\endcsname{\color{black}}%
      \expandafter\def\csname LT2\endcsname{\color{black}}%
      \expandafter\def\csname LT3\endcsname{\color{black}}%
      \expandafter\def\csname LT4\endcsname{\color{black}}%
      \expandafter\def\csname LT5\endcsname{\color{black}}%
      \expandafter\def\csname LT6\endcsname{\color{black}}%
      \expandafter\def\csname LT7\endcsname{\color{black}}%
      \expandafter\def\csname LT8\endcsname{\color{black}}%
    \fi
  \fi
    \setlength{\unitlength}{0.0500bp}%
    \ifx\gptboxheight\undefined%
      \newlength{\gptboxheight}%
      \newlength{\gptboxwidth}%
      \newsavebox{\gptboxtext}%
    \fi%
    \setlength{\fboxrule}{0.5pt}%
    \setlength{\fboxsep}{1pt}%
\begin{picture}(7200.00,4320.00)%
    \gplgaddtomacro\gplbacktext{%
      \csname LTb\endcsname
      \put(596,652){\makebox(0,0)[r]{\strut{}$2$}}%
      \csname LTb\endcsname
      \put(596,882){\makebox(0,0)[r]{\strut{}$3$}}%
      \csname LTb\endcsname
      \put(596,1111){\makebox(0,0)[r]{\strut{}$4$}}%
      \csname LTb\endcsname
      \put(596,1341){\makebox(0,0)[r]{\strut{}$5$}}%
      \csname LTb\endcsname
      \put(596,1570){\makebox(0,0)[r]{\strut{}$6$}}%
      \csname LTb\endcsname
      \put(596,1800){\makebox(0,0)[r]{\strut{}$7$}}%
      \csname LTb\endcsname
      \put(596,2029){\makebox(0,0)[r]{\strut{}$8$}}%
      \csname LTb\endcsname
      \put(596,2259){\makebox(0,0)[r]{\strut{}$9$}}%
      \csname LTb\endcsname
      \put(596,2488){\makebox(0,0)[r]{\strut{}$10$}}%
      \csname LTb\endcsname
      \put(596,2718){\makebox(0,0)[r]{\strut{}$11$}}%
      \csname LTb\endcsname
      \put(596,2947){\makebox(0,0)[r]{\strut{}$12$}}%
      \csname LTb\endcsname
      \put(596,3177){\makebox(0,0)[r]{\strut{}$13$}}%
      \csname LTb\endcsname
      \put(596,3406){\makebox(0,0)[r]{\strut{}$14$}}%
      \csname LTb\endcsname
      \put(596,3636){\makebox(0,0)[r]{\strut{}$15$}}%
      \csname LTb\endcsname
      \put(596,3865){\makebox(0,0)[r]{\strut{}$16$}}%
      \csname LTb\endcsname
      \put(596,4095){\makebox(0,0)[r]{\strut{}$17$}}%
      \csname LTb\endcsname
      \put(708,448){\makebox(0,0){\strut{}$3$}}%
      \csname LTb\endcsname
      \put(1015,448){\makebox(0,0){\strut{}$4$}}%
      \csname LTb\endcsname
      \put(1322,448){\makebox(0,0){\strut{}$5$}}%
      \csname LTb\endcsname
      \put(1628,448){\makebox(0,0){\strut{}$6$}}%
      \csname LTb\endcsname
      \put(1935,448){\makebox(0,0){\strut{}$7$}}%
      \csname LTb\endcsname
      \put(2242,448){\makebox(0,0){\strut{}$8$}}%
      \csname LTb\endcsname
      \put(2549,448){\makebox(0,0){\strut{}$9$}}%
      \csname LTb\endcsname
      \put(2855,448){\makebox(0,0){\strut{}$10$}}%
      \csname LTb\endcsname
      \put(3162,448){\makebox(0,0){\strut{}$11$}}%
      \csname LTb\endcsname
      \put(3469,448){\makebox(0,0){\strut{}$12$}}%
      \csname LTb\endcsname
      \put(3776,448){\makebox(0,0){\strut{}$13$}}%
      \csname LTb\endcsname
      \put(4082,448){\makebox(0,0){\strut{}$14$}}%
      \csname LTb\endcsname
      \put(4389,448){\makebox(0,0){\strut{}$15$}}%
      \csname LTb\endcsname
      \put(4696,448){\makebox(0,0){\strut{}$16$}}%
      \csname LTb\endcsname
      \put(5003,448){\makebox(0,0){\strut{}$17$}}%
      \csname LTb\endcsname
      \put(5309,448){\makebox(0,0){\strut{}$18$}}%
      \csname LTb\endcsname
      \put(5616,448){\makebox(0,0){\strut{}$19$}}%
      \csname LTb\endcsname
      \put(5923,448){\makebox(0,0){\strut{}$20$}}%
      \csname LTb\endcsname
      \put(6230,448){\makebox(0,0){\strut{}$21$}}%
      \csname LTb\endcsname
      \put(6536,448){\makebox(0,0){\strut{}$22$}}%
      \csname LTb\endcsname
      \put(6843,448){\makebox(0,0){\strut{}$23$}}%
      \csname LTb\endcsname
      \put(748,922){\makebox(0,0)[l]{\strut{}}}%
      \csname LTb\endcsname
      \put(1055,1151){\makebox(0,0)[l]{\strut{}}}%
      \csname LTb\endcsname
      \put(1362,1381){\makebox(0,0)[l]{\strut{}}}%
      \csname LTb\endcsname
      \put(1668,1610){\makebox(0,0)[l]{\strut{}}}%
      \csname LTb\endcsname
      \put(1975,1840){\makebox(0,0)[l]{\strut{}}}%
      \csname LTb\endcsname
      \put(2282,2069){\makebox(0,0)[l]{\strut{}}}%
      \csname LTb\endcsname
      \put(2589,2299){\makebox(0,0)[l]{\strut{}}}%
      \csname LTb\endcsname
      \put(2895,2299){\makebox(0,0)[l]{\strut{}}}%
      \csname LTb\endcsname
      \put(3202,2299){\makebox(0,0)[l]{\strut{}}}%
      \csname LTb\endcsname
      \put(3509,2299){\makebox(0,0)[l]{\strut{}}}%
      \csname LTb\endcsname
      \put(3816,2758){\makebox(0,0)[l]{\strut{}}}%
      \csname LTb\endcsname
      \put(4122,3217){\makebox(0,0)[l]{\strut{}}}%
      \csname LTb\endcsname
      \put(4429,3217){\makebox(0,0)[l]{\strut{}}}%
      \csname LTb\endcsname
      \put(4736,3217){\makebox(0,0)[l]{\strut{}}}%
      \csname LTb\endcsname
      \put(5043,3217){\makebox(0,0)[l]{\strut{}}}%
      \csname LTb\endcsname
      \put(5349,3217){\makebox(0,0)[l]{\strut{}}}%
      \csname LTb\endcsname
      \put(5656,3217){\makebox(0,0)[l]{\strut{}}}%
      \csname LTb\endcsname
      \put(5963,3217){\makebox(0,0)[l]{\strut{}}}%
      \csname LTb\endcsname
      \put(6270,3217){\makebox(0,0)[l]{\strut{}}}%
      \csname LTb\endcsname
      \put(6576,3217){\makebox(0,0)[l]{\strut{}}}%
      \csname LTb\endcsname
      \put(6883,3217){\makebox(0,0)[l]{\strut{}}}%
      \csname LTb\endcsname
      \put(3852,2821){\makebox(0,0)[l]{\strut{}Ex 7.3}}%
    }%
    \gplgaddtomacro\gplfronttext{%
      \csname LTb\endcsname
      \put(186,2373){\rotatebox{-270}{\makebox(0,0){\strut{}l(3,n)}}}%
      \csname LTb\endcsname
      \put(3775,142){\makebox(0,0){\strut{}n participants}}%
      \csname LTb\endcsname
      \put(2276,3912){\makebox(0,0)[r]{\strut{}best known}}%
      \csname LTb\endcsname
      \put(2276,3708){\makebox(0,0)[r]{\strut{}\cref{alg:koutofnphys}}}%
      \csname LTb\endcsname
      \put(2276,3504){\makebox(0,0)[r]{\strut{}\cref{alg:threeoutofn}}}%
      \csname LTb\endcsname
      \put(2276,3300){\makebox(0,0)[r]{\strut{}\cref{thm:2k}}}%
      \csname LTb\endcsname
      \put(2276,3096){\makebox(0,0)[r]{\strut{}\cref{cor:sperner}}}%
    }%
    \gplbacktext
    \put(0,0){\includegraphics[width={360.00bp},height={216.00bp}]{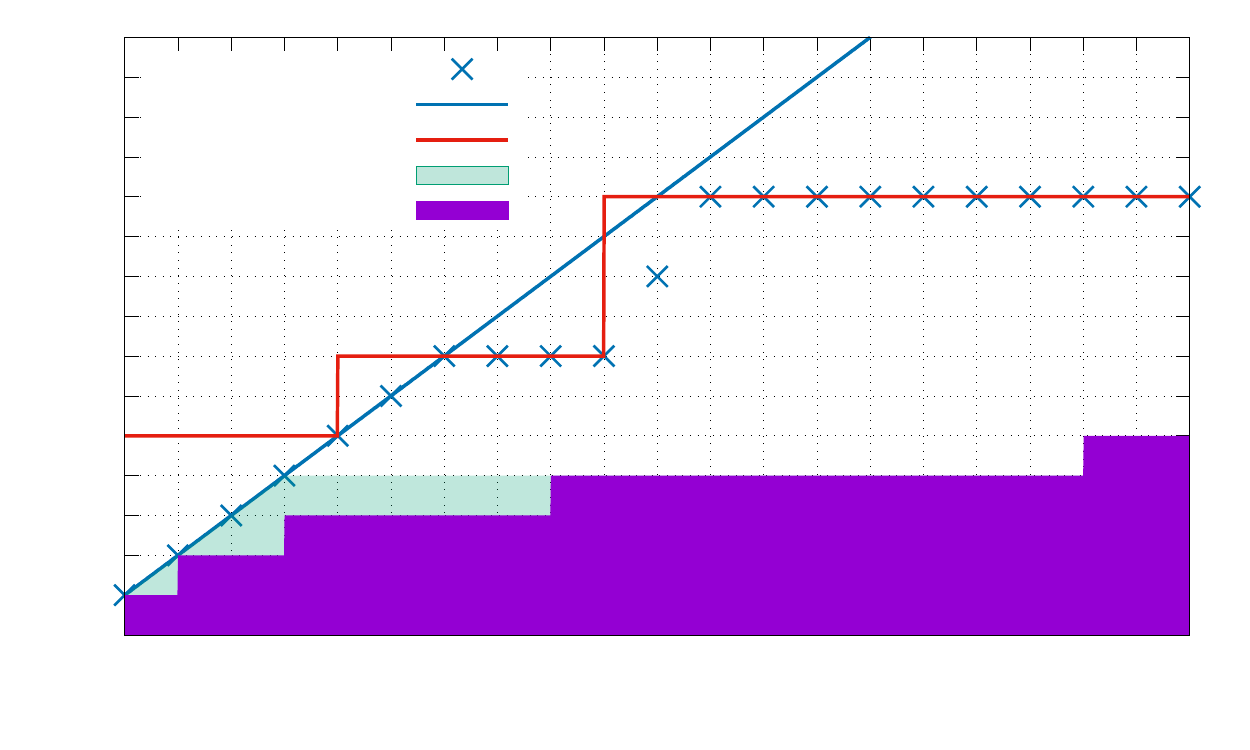}}%
    \gplfronttext
  \end{picture}%
\endgroup

%% file: recursive.tex
\section{A recursive asymptotic construction}\label{sec:rec}
For a larger number of participants, asymptotically, one can reduce
the number of padlocks by making subgroups.
For instance, consider building a $3$-out-of-$n$ system.
Create two subgroups $G$ and $H$ of
$n_0=\left\lceil\frac{n}{2}\right\rceil$
and
$n_1=\left\lfloor\frac{n}{2}\right\rfloor$
participants.
Setup a $3$-out-of-$n_0$ system for the participants of the subgroup
$G$.
Then duplicate all the distributed keys and give them to the members
of subgroup $H$, except potentially for the one supernumerary member of
subgroup $G$.
Then any $3$ participants all in one of the two subgroups can open the
system. Only some triads where one participant is in one subgroup,
and the two others in the other subgroup, cannot open the system yet.
But then, for these cases, we can build a conjunction of a
$1$-out-of-$n_i$ system with a
$2$-out-of-$n_{1-i}$ system.
Any triad of participants are either in a single subgroup or in a
one and two configuration and can open the system.
Now any single or pair of participants cannot open the $3$-out-of-$n_0$
system, nor both of the $1$-out-of-$n_i$ and $2$-out-of-$n_{1-i}$ systems.
Denote by $\PTS{3,n}$ the $3$-out-of-$n$ systems, we have thus
proven that:
\begin{equation}\label{eq:pts3}\begin{split}
\forall{n_0\geq{n_1}\geq{3}},n_0+n_1=n,
\text{OR}(~~ & \PTS{3,n_0};\\
 & \text{AND}\left(\PTS{1,n_0};\PTS{2,n_1}\right);\\
 & \text{AND}\left(\PTS{2,n_0};\PTS{1,n_1}\right)\\
)& \in\PTS{3,n}
\end{split}
\end{equation}

To count the number of keys and padlocks, we first need Faulhaber's
formula:
\begin{equation}\label{eq:faulhaber}
\sum_{j=1}^s j^k = \frac{1}{k+1} s^{k+1} +\BO{s^k}
\end{equation}

Then we need the following formula:
\begin{equation}\label{eq:binomial}\begin{split}
\sum_{k=0}^i (-1)^{i-k} \binom{i}{k}\frac{1}{i-k+1} &=
\frac{1}{i+1}\sum_{k=0}^i (-1)^{i-k} \binom{i+1}{k} \\
&= \frac{1}{i+1}\left(-(-1)^{-1}\binom{i+1}{i+1}+\sum_{k=0}^{i+1}
  (-1)^{i-k} \binom{i+1}{k} \right) \\
&=\frac{1}{i+1}(1+ (1-1)^{i+1})=\frac{1}{i+1}.
\end{split}
\end{equation}

Finally, we need the following
variant of the master theorem.
\begin{lemma}
  \label{lem:masterthm}
For $n\in\N$ and $i\in\N^*$, let $T(n)$ be a function defined by the recurrence $T(n)=T(\lceil{n/2}\rceil)+c\left(\log_2\left(\frac{n}{2}\right)\right)^i + \LO{\log(n)^i}$.
Then
$T(n)= \frac{c}{i+1}\log_2(n)^{i+1} + \LO{\log(n)^{i+1}}$.
\end{lemma}
\begin{proof} Expanding $log_2(n)$ times the recurrence, we obtain:
  \begin{equation}
    \begin{split}
      T(n)
      &= T(1)+ \sum_{j=1}^{\log_2 n} c\left(\log_2\left(\frac{n}{2^{j+1}}\right)\right)^i+ \LO{\left(\log\left(\frac{n}{2^{j}}\right)\right)^i}\\
      &= c\sum_{j=1}^{\log_2 n}\left( \sum_{k=0}^i (-1)^{i-k}\binom{i}{k}
      \log_2(n)^k j^{i-k}\right) + \LO{\log(n)^i}\\
       &= c\sum_{k=0}^i\left( (-1)^{i-k}\binom{i}{k}
      \log_2(n)^k \sum_{j=1}^{\log_2 n} j^{i-k}\right) + \LO{\log(n)^i}
    \end{split}
  \end{equation}
Using~\Cref{eq:faulhaber}, with $s=\log_2(n)$, this is:
\begin{equation}\begin{split}
  T(n) &= c\sum_{k=0}^i\left( (-1)^{i-k}\binom{i}{k}
      \log_2(n)^k \frac{1}{i-k+1}\log_2(n)^{i-k+1}\right) + \BO{\log(n)^i}\\
      &= c\log_2(n)^{i+1} \sum_{k=0}^i\left(
      (-1)^{i-k}\binom{i}{k}\frac{1}{i-k+1}\right) + \BO{\log(n)^i}
\end{split}
\end{equation}
Finally, with \Cref{eq:binomial}, we have that:
\begin{equation}\begin{split}
  T(n) &= \frac{c}{i+1}\log_2(n)^{i+1} + \LO{\log(n)^{i+1}}
\end{split}
\end{equation}
\end{proof}

With these, we can now count padlocks and keys for the strategy with
two subgroups of \Cref{eq:pts3}:
\begin{lemma}\label{lem:3log} For $n\geq{6}$,
\[\ell_{3,n}\leq{2\log_2(n)^2+\LO{\log(n)^2}}\]
and the upper bound is attained with an average of
$\frac{1}{2}\log_2(n)^2+\LO{\log(n)^2}$ keys per participant.
\end{lemma}
\begin{proof}
To realize \Cref{eq:pts3} we need
1 padlock for \PTS{1,n_0} and another one for \PTS{1,n_1}.
We also need less than $2\lceil\log_2(n_0)\rceil$ padlocks
for \PTS{2,n_0} and similarly
$2\lceil\log_2(n_0)\rceil$ padlocks
for \PTS{2,n_1}, using \cref{thm:twoi}.
finally, \PTS{3,n_0} is realized recursively.
Therefore the number of padlocks for \PTS{3,n}
satisfies
$P_3(n)\leq P_3(n_0)+2+4\log_2(n/2)+\BO{1}$.
\cref{lem:masterthm} then gives
$P_3(n)=\frac{4}{2}\log_2(n)^2+\LO{\log(n)^2}$.

Similarly the participants of subgroup $G$ get $1$ key for
\PTS{1,n_0} and $\log_2(n_0)$ keys for \PTS{2,n_0}.
The participants in the other subgroup get
$1$ key for \PTS{1,n_1} and
 $\log_2(n_1)$ keys for \PTS{2,n_1}.
Then they each get the keys needed for $\PTS{3,n_0}$.
Thus the average number of keys per participant satisfies
$K_3(n)\leq K_3(n/2)+1+\log_2(n/2)+\BO{1}$.
\cref{lem:masterthm} then gives
$K_3(n)=\frac{1}{2}\log_2(n)^2+\LO{\log(n)^2}$.
\end{proof}

Now, this scheme can be generalized for any threshold $k$ as shown in
\cref{alg:recursive}.

\begin{algorithm}[htb]
\caption{Recursive $k$-out-of-$n$ threshold system with shared keys}\label{alg:recursive}
\begin{algorithmic}[1]
\Require $n\geq{k\geq{2}}$.
\Ensure A recursively build $k$-out-of-$n$ threshold padlock system.
\If{$k==2$}
\State\Return \cref{alg:twooutofn}.
\ElsIf{$k(k+1)/2\geq{n}$}
\State\Return a $k$-out-of-$n$ system with $n$ padlocks.\Comment{\cref{thm:2k}}
\Else
\State Let $n_0=\left\lceil\frac{n}{2}\right\rceil$
and $n_1=\left\lfloor\frac{n}{2}\right\rfloor$;
\State Separate the participants in two groups $G$ and $H$ with $n_0$ and $n_1$ members;
\State Recursively setup a $k$-out-of-$n_0$ system;
\State Duplicate the keys of this system and distribute one set to
members of $G$ and the other set to members of $H$;
\State Setup an $1$-out-of-$k$ OR system without padlocks nor keys;
\State Attach the $k$-out-of-$n_0$ system to one of the latches of the
OR system;
\For{$i=1..k-1$}
\State Setup a $2$-out-of-$2$ AND system and attach it to the OR system
\State Setup a $i$-out-of-$n_0$ system for members of the subgroup $G$
and attach it to this AND system;
\State Setup a $(k-i)$-out-of-$n_1$ system for members of the subgroup
$H$ and attach it to this AND system;
\EndFor
\State\Return the OR system openable either via the $k$-out-of-$n_0$
system or via one of the $(k-1)$ AND Systems.
\EndIf
\end{algorithmic}
\end{algorithm}

\begin{theorem}\label{thm:recursive}
\cref{alg:recursive} is correct
and asymptotically requires
\begin{equation}
\begin{cases}
\frac{2^{k-1}}{(k-1)!}\log_2(n)^{k-1}+\LO{\log(n)^{k-1}}~\text{padlocks}\\
\frac{1}{(k-1)!}\log_2(n)^{k-1}+\LO{\log(n)^{k-1}}~\text{keys per participants}\\
\end{cases}
\end{equation}
\end{theorem}
\begin{proof}
  For the correctness, consider a group of at most $k-1$ participants.
  They cannot open the $k$-out-of-$n_0$ system. Then they are
  distributed with $j\in{0..k-1}$ of them in group $G$ and $k-1-j$ in
  group $H$. They can thus open any of the \PTS{\alpha,n_0} for $\alpha=0..j$,
  but none of the corresponding \PTS{k-\alpha,n_1} since $k-\alpha>k-1-j$.
  They can also open any of the \PTS{\beta,n_1} for $\beta=0..k-1-j$,
  but none of the corresponding \PTS{k-\beta,n_0} since $k-\beta>j$.
  So they can never open the system.
  On the contrary, consider a group of at least $k$ participants.
  They are distributed with $j\in{0..k}$ in group G and at least $k-j$
  in group H. Thus they can either open the $k$-out-of-$n_0$ system or
  one of the \PTS{j,n_0} AND \PTS{k-j,n_1} group.

  Now, for the complexity bound, we proceed by induction on
  $k\leq{n}$. The formulae are correct for $k=3$ by \cref{lem:3log}.
  Now suppose that the formulae are correct $\forall{i\leq{k}}$
  and consider \cref{alg:recursive} at $k+1$.
  Then the number of padlocks used by the Algorithm is
  $P_{k+1}(n)=P_{k+1}(n_0)+\sum_{i=1}^{k}P_i(n_0)+P_{k+1-i}(n_1)$
  and the average number of keys per participant is
  $K_{k+1}(n)=K_{k+1}(n_0)+\frac{1}{2}\left(\sum_{i=1}^{k}K_i(n_0)+K_{k+1-i}(n_1)\right)$.
  By the induction hypothesis, the number of padlocks thus
  satisfy:
  \begin{equation}\begin{split}
      P_{k+1}(n) \leq & P_{k+1}(n/2)+\sum_{i=1}^k\left(
      \frac{2^{i-1}}{(i-1)!}\log_2(n/2)^{i-1} +
      \frac{2^{k-i}}{(k-i)!}\log_2(n/2)^{k-i}\right)\\
      &  + \LO{\log(n)^{k-1}}\\
      = &P_{k+1}(n/2) + 2\frac{2^{k-1}}{(k-1)!}\log_2(n/2)^{k-1}
      + \LO{\log(n)^{k-1}};
\end{split}\end{equation}
  and the number of keys satisfies:
\begin{equation}\begin{split}
    K_{k+1}(n)\leq &K_{k+1}(n/2)+\frac{1}{2}\left(\sum_{i=1}^k
      \frac{1}{(i-1)!}\log_2(n/2)^{i-1}\right.\\
    & + \left.\frac{1}{(k-i)!}\log_2(n/2)^{k-i}\right)
    + \LO{\log(n)^{k-1}}\\
    = & K_{k+1}(n/2) + \frac{2}{2}\frac{1}{(k-1)!}\log_2(n/2)^{k-1}
    + \LO{\log(n)^{k-1}}.
\end{split}\end{equation}
  Finally, \cref{lem:masterthm}, applied on both relations shows that:
\begin{align}
  P_{k+1}(n)&=\frac{2^k}{(k-1)!k}\log_2(n)^{k-1+1}+\LO{\log(n)^k}\\
  K_{k+1}(n)&=\frac{1}{(k-1)!k}\log_2(n)^{k-1+1}+\LO{\log(n)^k}
\end{align}
  These establish that the hypothesis is true for $k+1$. Therefore it is
  inductive and the theorem is proven.
\end{proof}

Note that \cref{alg:recursive} is useful only for a large number of
participants. For instance with a threshold of three,
$\lceil{2\log_2(n)^2}\rceil$ is lower than  \cref{thm:bose} only for
$n\geq{33\,922}$.
This is overestimated, but, more precisely, the smallest case where
\cref{alg:recursive} yields less padlocks than \cref{alg:threeoutofn}
is only at $n=1248$. There, we have $63$ padlocks for a
$3$-out-of-$624$ system with \cref{alg:threeoutofn},
then $12$ padlocks for a $2$-out-of-$624$ system via
\cref{alg:twooutofn}. This is a total of $63+2(1+12)=89$ padlocks
where \cref{alg:threeoutofn} alone yields $91$ padlocks.
Now, for $k=4$, and using \cref{thm:bose} when $i=3$,
the smallest case where \cref{alg:recursive} yields less than $n$
padlocks is at $n=114$.
There, we have $57$ padlocks for a $4$-out-of-$57$
system,
then $8$ padlocks for a $2$-out-of-$57$ system via
\cref{alg:twooutofn} and $19$ padlocks for a $3$-out-of-$57$ system
via \cref{alg:threeoutofn}. This is a total of $57+2(1+8+19)=113$ padlocks
for a $4$-out-of-$114$ system via \cref{alg:recursive}.

%% file: conclusion.tex
We designed a physical $k$-out-of-$n$ threshold lock that can be used
for various applications, including physical access control, voting or
secret sharing.
Our system only uses $n$ padlocks, showing that previous exponential
answers to Liu's problem were far too pessimistic.
For $k=2$, we were even able to identify an optimal
solution using our device, which 
needs less than $2\lceil\log_2(n)\rceil$ padlocks, but requires duplicating keys.
We also show that for $k\geq\sqrt{2n}$ the minimal
number of padlocks is $n$ (and our device also reaches this).

There are many open questions left, for example 
we have shown that reducing the number of padlocks is equivalent to
reducing the size of the fields for interpolation-based secret
sharing, but further exploration of the links with digital systems could be
envisioned.
Another future work is to find minimal solutions in terms of
padlocks for small cases, in particular for $k$ between $3$ and
$\sqrt{2n}$. For instance, when $k\geq{3}$, Johnson's bound suggests
that it might be possible to build systems with only ${\mathcal
O}(k\sqrt{n})$ padlocks and we were able to prove this for $k=3$.

We also devised algorithms
that can implement more complex access policies beyond simple
thresholds, expressed as disjunctive or conjunctive Boolean formulas.
It is yet unclear for us whether there are general solutions using
less locks than the number of variables. 

We proposed one variant using sealed wire and wrappings to provide
an alternative solution to our device with exactly $n$ padlocks.
The threshold systems we found with this approach unfortunately use an
exponential number of wrappings. It is unclear to us if this could be
improved.

Differently, on the asymptotic side, we have found an algorithm,
recursively combining several of our devices,
requiring only $\BO{log(n)^{k-1}}$ padlocks for $k$-out-of-$n$
threshold systems but we have only a lower bound of $\BO{log(n)}$.

Finally, if we do not only count the number of padlocks, but more
generally the number of keys or of latches, then clearly a lower bound
on the number of devices is $n$: each player must at least have
something. Otherwise groups of $k$ players with an empty player would
have the same abilities of a group of $k-1$ players. 
With this model of complexity, our $k$-out-of-$n$ designs are
asymptotically optimal as they require just $n$ padlocks, $n$ latches
and $n$ keys.